\documentclass[9pt,twocolumn,twoside]{IEEEtran} 

\IEEEoverridecommandlockouts                        
\usepackage{times,color,graphicx,setspace,bbm,mathdots,mathrsfs} \usepackage{amssymb,latexsym,amsfonts,amsmath,cite,stmaryrd,caption, comment,relsize}
\usepackage[dvips]{psfrag}
\usepackage{setspace}
\usepackage{url}
\usepackage{pgf,tikz}
\usepackage{comment}
\usepackage[linesnumbered,ruled]{algorithm2e} 

\newtheorem{theorem}{Theorem}
\newtheorem{proof}{Proof}
\newtheorem{lemma}{Lemma}
\newtheorem{corollary}{Corollary}
\newtheorem{proposition}{Proposition}
\newtheorem{remark}{Remark}

\newtheorem{example}{Example}
\newtheorem{assumption}{Assumption}
\newtheorem{definition}{Definition}
 \newcommand{\R}{{\mathbb{R}}}

\newcommand{\G}{{\mathcal{G}}}
\newcommand{\HH}{{\mathcal{H}}}
\newcommand{\V}{{\mathcal{V}}}
\newcommand{\EE}{{\mathcal{E}}}

\newcommand{\rhoo}{{\rho}}

\newcommand{\NN}{\mathfrak{N}}

\newcommand{\LL}{{\mathfrak{L}_{n}}}

\newcommand{\ie}{{i.e.,}}

\newcommand{\Rel}{{\text{Re}}}
\newcommand{\diag}{{\text{diag}}}
\usepackage{pgf,tikz}

\DeclareMathOperator{\tr}{{Tr}}
\DeclareMathOperator{\E}{\mathbb{E}}

\newcommand{\RC}{\color{red}}
\newcommand{\BC}{\color{blue}}

\begin{document}

\title{Abstraction of Linear Consensus Networks with Guaranteed  Systemic Performance Measures } 
\author{Milad Siami$^{\dagger}$  
        and~Nader Motee$^{\star}$

\thanks{$^{\dagger}$ M. Siami is with the Institute for Data, Systems, and Society, Massachusetts Institute of Technology, Cambridge, MA 02319. Email:  {\tt\small siami@mit.edu}.}
\thanks{$^{\star}$ N. Motee is with the Department of Mechanical Engineering and Mechanics, Packard Laboratory, Lehigh University, Bethlehem, PA 18015. Email:  {\tt\small motee@lehigh.edu}.}}
\maketitle

\begin{abstract}   

A proper abstraction of a large-scale linear consensus network with a dense coupling  graph is one whose number of coupling links is proportional to its number of subsystems and its performance is comparable to the original network. Optimal design problems for an abstracted network are more amenable to efficient optimization algorithms. From {the} implementation point of view, maintaining such networks are usually more favorable and cost effective due to their reduced communication requirements across a network. Therefore, approximating a given dense linear consensus network by a suitable abstract network is an important analysis and synthesis problem. In this paper, we develop a framework to compute {an} abstraction of a given large-scale linear consensus network with guaranteed performance bounds using a nearly-linear time algorithm.
First, the existence of abstractions of a given network is proven. Then, we present an efficient and fast algorithm for computing a proper abstraction of a given network. 
Finally, we illustrate {the} effectiveness of our theoretical findings via several numerical simulations. 
\end{abstract}

\section{Introduction}
\allowdisplaybreaks

Reducing design complexity in interconnected networks of dynamical systems by means of abstraction are central in several real-world applications \cite{Siami14cdc-2, siami14acc, Olfati, Siami13siam, farlinjovTAC14sync,Summers16}. Various notions of abstractions for dynamical systems have been widely used by researchers in the context of control systems in past  decades, see \cite{Girard, H2error, Model_reduction, Rantzer} and references in there, where the notion of {\it reduction} mainly implies projecting dynamics of a system to lower dimensional state spaces. In this paper, we employ a relevant notion of abstraction in the context of interconnected dynamical network: for a given dynamical network that is defined over a coupling graph, find another dynamical system whose coupling graph is significantly sparser and its performance quality remains close to that of the original network. In this definition, abstraction can be regarded as a notion of network reduction. There are several valid reasons why reduction in this sense is useful in design, maintenance, and implementation of dynamical networks. Real-time generation of state estimation in large-scale dynamical networks can be done much more efficiently and faster if proper abstractions are utilized. Optimal control problems that involve controller design, feedback gain adjustments, rewiring existing feedback loops, and etc. are more amenable to efficient computational tools that are specifically tailored for optimization problems with sparse structures. In security- or privacy-sensitive applications such as formation control of group of autonomous drones, it is usually required to minimize communication requirements across the network to reduce risk of external intrusions. In power network applications,  network authorities periodically provide access to their network data and parameters for academic (or public)  studies and evaluations. In order to reduce possibility of planned malicious attacks, network authorities can perform abstractions in order to hide actual values of parameters in their networks by preserving all other important characteristics of the network that interest researchers. 

The goal of this paper is to address the abstraction problem for the class of linear consensus networks. In \cite{Siami14cdc-2}, we introduce a class of operators, so called systemic performance measure, for linear consensus networks that provides a unified framework for network-wide performance  assessment.  Several existing and popular performance measures in the literature, such as $\mathcal H_2$ and $\mathcal H_{\infty}$ norms of a consensus network from a disturbance input to its output, are examples of systemic performance measures. {This class of operators is obtained through our close examination of functional properties of several
existing gold standard measures of performance in the context of network engineering and
science. }
An important contribution of this reference paper is that it enables us to optimize performance of a consensus network solely based on its intrinsic features. The authors formulate several optimal design problems, such as weight adjustment as well as rewiring of coupling links, with respect to this general class of systemic performance measures and propose efficient algorithms to solve them.  In \cite{Siami13cdc,Siami14arxiv}, we quantify several fundamental tradeoffs between a $\mathcal{H}_2$-based performance measure and  sparsity measures of a linear consensus network. The problem of sparse consensus network design  has been considered before in \cite{farlinjovTAC14sync,wujovACC14,mogjovACC15,Dhingra12}, where they formulate an $\ell_0$-regularized  $\mathcal H_2$ optimal control problem.   The main common shortcoming of existing works in this area is that they are heavily relied on computational tools with no analytical performance guarantees for the resulting solution. More importantly, the proposed methods in these papers mainly suffer from high computational complexity as network size grows.  

For a given linear consensus network with an undirected connected graph, the  network abstraction problem seeks to construct a new network with a reasonably sparser graph compared to the original network such that the dynamical behavior of the two networks remains similar in an appropriately defined sense. We develop a methodology that computes abstractions of a given consensus network using a  nearly-linear time $\tilde{\mathcal O}(m)$\footnote{We use $\tilde {\mathcal O}(.)$ to hide poly $\log \log$ terms from the asymptotic bounds. Thus, $ f(n)\in {\tilde {\mathcal O}}\left(g(n)\right)$ means that there exists $c>0$ such that $f(n)\in \mathcal O\left (g(n)\log ^{c}g(n)\right )$.} algorithm with guaranteed systemic performance bounds, where $m$ is the number of links. Unlike other existing work on this topic in the literature, our proposed framework:  (i) works for a broad class of systemic performance measures including $\mathcal H_2$-based performance measures, (ii) does not involve any sort of relaxations such as $\ell_0$ to $\ell_1$,\footnote{ We discuss some of the shortcomings of the $\ell_0$/$\ell_1$-regularization based sparsification methods in Section \ref{ell0-ell1}.} (iii) provides guarantees for the existence of a sparse solution, (iv) can partially sparsify predetermined portions of a given network; and most importantly, (v) gives guaranteed levels of performance.    While our approach is relied on several existing works in algebraic graph theory \cite{Spielman, Spielman14}, our control theoretic contributions are threefold. First,  we show that there exist proper  abstractions for every given linear consensus network. Second,  we develop a framework to compute a proper abstraction of a network using a fast randomized algorithm. One of the main features of our method is that while the coupling graph of the abstracted network is a subset of the coupling graph of the original network, the link weights  (the strength of each coupling) in the  sparsified network are adjusted accordingly to reach predetermined levels of systemic performance. Third, we prove that our method can also be applied for partial abstraction of large-scale networks, which means that we can abstract a prespecified subgraph of the original network. This is practically plausible as our algorithm can obtain an abstraction using only spatially localized  information. Moreover, this allows parallel implementation of the abstraction algorithm in order to achieve comparably lower time complexity.

\section{Notation and Preliminaries}
\allowdisplaybreaks

	The set of real, positive real, and strictly positive real numbers are represented by $\R$, $\R_+$ and $\R_{++}$, respectively. A matrix is generally represented by an upper case letter, say $X = [x_{ij}]$, where $x_{ij}$ is the $(i,j)$ element of matrix $X$ and $X^{\text T}$ indicates its transposition. We assume that $\mathbbm{1}_n$ and $I_n$ denote the $n \times 1$ vector of all ones and the $n \times n$ identity matrix, respectively. The centering matrix is defined by $M_n=I_n - \frac{1}{n} J_n$ in which $J_n$ is the $n \times n$ matrix of all ones. {Notation $X\succeq Y$ is equivalent to matrix $X-Y$ being positive semi-definite.}{ A graph is represented by $\G~=~(\V, \EE,w)$, where $\V$ is the set of nodes, $\EE \subset \V \times \V$ is the set of links, and $w: \V \times \V \rightarrow  \R_{+}$ is the weight function. The value of the weight function is zero for $e \in \V \times \V \, \backslash \, \EE $ and positive for $e \in \EE$.} { The weighted degree of  node $i \in \V$ is defined by
	\begin{equation}
		d_i ~:=~\sum_{e=\{i,j\} \in \EE} w(e).
	\end{equation}
The neighborhood of node $i$ is denoted by set $\mathcal N (i)$ that consists of all adjacent nodes to $i$ and its cardinality $|\mathcal N (i)|$ is equal to the number of neighbors of node $i$. In unweight graphs, $|\mathcal N (i)|$ is equal to the degree of node $i$.}
The adjacency matrix $A = [a_{ij}]$ of graph $\G$ is defined by setting $a_{ij} = w(e)$ if $e=\{i,j\} \in \EE$, and $a_{ij}=0$ otherwise. The Laplacian matrix of graph $\G$ with $n$ nodes is defined by \[L :=\diag[d_{1},\ldots,d_{n}] - A.\] 
A  $n$-by-$m$ oriented incidence matrix $E=[\mathfrak{e}_{ij}]$ for $1 \leq i \leq n$ and $1 \leq j \leq m$ can be formed by assigning an arbitrary direction for every link of $\G$, labeling every link by a number $1 \leq j \leq m$, and letting $\mathfrak{e}_{ij} = 1$ whenever node $i$ is the head of  (directed) link $j$, $\mathfrak{e}_{ij}=-1$ if node $i$ is the tail of  (directed) link $j$, and $\mathfrak{e}_{ij}=0$ when link $j$ is not attached to node $i$ for all possible orientations of links. The weight matrix $W=[w_{kk}]$ is the $m$-by-$m$ diagonal matrix with diagonal elements $w_{kk}=w(e_k)$ for $1 \leq k \leq m$. It follows that
	\begin{equation*}
		L ~=~ EWE^{\text T}.
	\label{L and E and W}
	\end{equation*}
\begin{assumption}\label{assum-on-graphs}
All graphs in this paper are assumed to be finite, simple, undirected, and connected.
\end{assumption}	
According to this assumption, every considered Laplacian matrix in this paper has exactly $n-1$ positive eigenvalues and one zero eigenvalue, which allow us to index them in ascending order $0=\lambda_1 \leq \lambda_2 \leq \cdots \leq \lambda_n$. The set of Laplacian matrices of all connected weighted graphs over $n$ nodes is represented by $\LL$. The Moore-Penrose pseudo-inverse of $L$ is denoted by $L^{\dag}=[l_{ji}^{\dag}]$ which is a square, symmetric, doubly-centered and positive semi-definite matrix. The corresponding resistance matrix $R=[r_{ij} ]$ to Laplacian matrix $L$ is defined by setting 
\[r_{ij} \, = \, l_{ii}^{\dag}+l_{jj}^{\dag}-2l_{ij}^{\dag}\]
in which $r_{ij}$ is called the effective resistance between nodes $i$ and $j$. Moreover, we denote the effective resistance of link $e=\{i,j\}$ by $r(e)=r_{ij}=r_{ji}$. The $\ell_0$ sparsity measure of matrix $A=[a_{ij}] \in \R^{n \times n}$ is defined by 
\begin{equation}
\|A\|_{\ell_0} := \mathrm{card} \big\{ (i,j) ~\big|~a_{ij} \neq 0 \big\}. \label{sparsity-measure-1}
\end{equation}
The $\mathcal{S}_{0,1}$ sparsity measure of matrix $A$ is defined by 
\begin{equation}
\|A\|_{\mathcal{S}_{0,1}} := \max \Big\{ \max_{1 \leq i \leq n} \| A(i,.) \|_{0}, \max_{1 \leq j \leq n} \| A(.,j) \|_{0}  \Big\},\label{local-measure}
\end{equation}
where $A(i,.)$ represents the $i$'th row and $A(.,j)$ the $j$'th column of matrix $A$.  The value of the $\mathcal{S}_{0,1}$-measure of a matrix is the maximum number of nonzero elements among all rows and columns of that matrix \cite{motee-sun}.
\section{Problem Statement}

\subsection{Network model}
We consider a class of consensus networks that consist of  a group of subsystems whose state variables $x_i$,  control inputs $u_i$, and output variables $y_i$ are scalar and their dynamics evolve with time according to 
\begin{eqnarray}
\dot{x}_i(t) & = & u_i(t) +\xi_i(t) \label{TI-consensus-algorithm} \\
y_i(t) & = & x_i(t) - \bar{x}(t)  \label{TI-consensus-algorithm-2}
\end{eqnarray}
for all $i=1,\ldots,n$, where {$x_i(0)=x_0^i$} is the initial condition and \[\bar{x}(t)=\frac{1}{n}\big(x_1(t)+\ldots+x_n(t)\big)\] 
is the average of all states at time instant $t$. The impact of the uncertain environment on each agent's dynamics is modeled by the exogenous noise/disturbance input $\xi_i$. By applying the following linear feedback control law to the agents of this network 
\begin{equation}
u_i(t) ~=~\sum_{j=1}^{n} k_{ij} \big(x_j(t) - x_i(t)\big),\label{feedback-law}
\end{equation}
where $k_{ij}$ is the feedback gain between subsystems $i$ and $j$, the closed-loop dynamics of network (\ref{TI-consensus-algorithm})-(\ref{feedback-law}) can be written in the following compact form
	\begin{equation}
		\NN(L):
		\begin{cases}
			\dot x(t)~=~-L x(t)+\xi(t)\\
			y(t)~=~M_n x(t)
		\end{cases}
	\label{first-order-G}
	\end{equation}
with initial condition $x(0) ~=~ x_0$, where $x$, $\xi$ and $y$ denote the state vector of the entire network, the exogenous disturbance input and the output vector of the network, respectively. The Laplacian matrix $L=[l_{ij}]$ is defined by 
\begin{equation}
\displaystyle l_{ij} := \left\{\begin{array}{ccc}
-k_{ij} & \textrm{if} & i \neq j \\
 &  &  \\
k_{i1}+\ldots+k_{in}& \textrm{if} & i=j
\end{array}\right..
\end{equation}
%
%
The coupling graph of the consensus  network \eqref{first-order-G} is a graph $\G=(\V,\mathcal E, w)$ with node set $\V=\{1,\ldots,n\}$, link set 
\begin{equation} 
	\EE=\Big\{ \{i,j\}~\Big|~\forall~i,j \in \V:~k_{ij} \neq 0\Big\}, \label{edge-set}
\end{equation}
and weight function 
\begin{equation}
	w(e)= \left\{\begin{array}{ccc}
k_{ij} & \textrm{if} & e=\{i,j\} \in \EE \\
 &  &  \\
0& \textrm{if} & e \notin \EE
\end{array}\right..\label{edge-weight}
\end{equation}
One may verify that the Laplacian matrix of graph $\G$ is equal to $L$. 
\begin{assumption}\label{assump-simple}
All feedback gains (weights) satisfy the following properties for all $i,j \in \V$: 

\vspace{0.1cm}
\noindent (i)~non-negativity: $k_{ij} \geq 0$, \\
\noindent (ii)~symmetry: $k_{ij}=k_{ji}$,\\
\noindent (iii)~simpleness: $k_{ii}= 0$.
\vspace{0.1cm}
\end{assumption}

Property (ii) implies that feedback gains are symmetric and (iii) means that there is no self-feedback loop in the network.

\vspace{0.1cm}
\begin{assumption}\label{assum-coupling-graph}
The coupling graph $\G$ of the consensus network \eqref{first-order-G} is  time-invariant.
\end{assumption}
\vspace{0.1cm}

Based on Assumption \ref{assum-coupling-graph},  the corresponding eigenvector to the only marginally stable mode of the network is $\mathbbm{1}_n$. This mode is unobservable from the performance output as the output matrix of the network satisfies $M_n \mathbbm{1}_n=0$.

\subsection{Homogeneous Systemic Performance Measures}

	\begin{figure}[t]
	\centering
	 \psfrag{A}[h][h]{ \footnotesize{$\Omega$}}  
	 \psfrag{B}[t][t]{ \footnotesize{$\Omega_s$}}  
	 \psfrag{C}[t][t]{ \footnotesize{$\Omega_h$}}        
	\includegraphics[trim = 0 0 0 0, clip,width=.2 \textwidth]{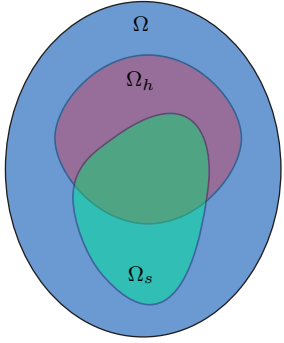} 
	\caption{\small A Venn diagram that shows the relationship among sets $\Omega$, $\Omega_s$, and $\Omega_h$.  The  set of general systemic measures $\Omega$ is a superset of both the set of homogeneous systemic measures $\Omega_h$ and the set of spectral systemic measures $\Omega_s$. While the intersection of sets $\Omega_s$ and $\Omega_h$ is nonempty, there are some systemic measures that belong only to one of these sets.  %
	}
	\label{fig:venn}
\end{figure}

	A {systemic measure} in this paper refers to a real-valued operator over the set of all consensus networks governed by  (\ref{first-order-G}) with the purpose of quantifying performance of this class of networks in presence exogenous uncertainties.   Since every network with dynamics (\ref{first-order-G}) is uniquely determined by its Laplacian matrix, it is reasonable to define a systemic performance measure as an operator on set $\LL$. 
\begin{definition}\label{def-homo-systemic}
	An operator $\rhoo: \mathfrak L_n \rightarrow \R_+$ is called  a  homogeneous systemic measure of {order $-\alpha$, where $\alpha > 0$}, if it satisfies the following properties for all matrices in $\mathfrak L_n$:	
	
	\noindent {1.} {\it Homogeneity:} For  all $\kappa >1$, 
	\[ \rhoo (\kappa L)~=~\kappa^{-\alpha} \rhoo (L);\]
	\noindent {2.} {\it Monotonicity:} If $L_{2} \preceq L_{1}$, then
				\[\rhoo (L_{1}) ~\leq~ \rhoo (L_{2});\]	
	\noindent {3.} {\it Convexity:} For all $0 \leq c \leq 1$,
	\[\rhoo (c L_{1}+(1-c)L_{2})~\leq~ c\rhoo (L_{1})+{(1-c)}\rhoo (L_{2});\]
	\end{definition}

The set of all homogeneous systemic performance measures is denoted by $\Omega_h$. We adopt an axiomatic approach to introduce and categorize a general class of performance measures that captures the quintessence of a meaningful measure of performance in
large-scale dynamical networks \cite{Chen2013axiomatic}. Property 1 implies that intensifying the coupling weights by ratio $\kappa>1$ results in $\kappa^{\alpha}$ times better performance. Property 2 guarantees that strengthening  couplings in a consensus network never worsens the network performance with respect to a given systemic measure. The monotonicity property induces a partial ordering on all linear consensus networks with dynamics  (\ref{first-order-G}). 
Adding new coupling links or strengthening the existing couplings will result in  better performance.  {Property 3 is imposed for the pure purpose of having favorable (convex) network design optimization problems.}

The class of systemic performance measures can be classified based on their functional properties according to Definition \ref{def-homo-systemic}. Let us denote the set of spectral systemic performance measures by $\Omega_s$. This class consists of all measures that satisfy properties 2, 3 and orthogonal invariance\footnote{A systemic measure is orthogonally invariant if $\rhoo(L) ~=~ \rhoo(ULU^{\text T})$ for every orthogonal matrix $U$ for which $UU^{\text T} = U^{\text T} U =I$.}. We refer to \cite{SiamiSystemic} for a comprehensive study of this class of performance measures. It is proven that all measures in $\Omega_s$ depend only on   Laplacian eigenvalues. Let us represent the set of all general systemic performance measures that only satisfy  properties 2 and 3 by $\Omega$. Fig. \ref{fig:venn} shows the relationship between the sets of spectral, homogeneous, and general systemic performance measures.

{
\begin{definition}\label{norm-perf}
For a given linear consensus network $\NN(L)$  endowed with a  homogeneous systemic measure $\rho: \LL \rightarrow \R_+$ of order $-\alpha$, its corresponding normalized performance index is defined by
\begin{equation}
\Pi_{\rho}(L) := \sqrt[\alpha]{\rho(L)}.\label{hemogeniz}
\end{equation}
\end{definition}

}

\subsection{Network Abstraction Problem}\label{subsec-problem}

Our goal is to develop a framework to compute an {\it abstraction} of a given linear consensus network with predetermined levels of performance  and sparsity (i.e., link reduction). 

\begin{definition}\label{abstract-def}
Let us consider network $\NN(L)$ that is governed by \eqref{first-order-G}. For a properly chosen pair of design parameters $d \in \R_{++}$ and $\epsilon \in (0,1)$, another  $\NN(L_s)$ is said to be an {$(\epsilon,d)$-abstraction} of $\NN(L)$ if and only if 

\vspace{0.1cm}
\noindent (i)  $\NN(L_s)$ has at most $dn/2$ feedback links;

\noindent (ii) $\NN(L_s)$ is an $\epsilon$-approximation of $\NN(L)$ in the following sense 
  	\begin{equation}\label{rho-approx}	
	\left| \frac{\Pi_{\rho}(L)-\Pi_{\rho}(L_s)}{\Pi_{\rho}(L_s)}	\right|\leq \epsilon			\end{equation}  
for every homogeneous systemic {performance} measure $\rho: \LL \rightarrow \R_+$.	
\end{definition}

{Property (i) implies that {the average number of neighbors for every node in $\NN(L_{s})$ is less than $d$, i.e., 
\[\bar{d}~=~\frac{1}{n}\sum_{i=1}^n |\mathcal N(i)| ~=~ 2 \frac{|\EE_s|}{n} ~\leq~ d,\]
where $\mathcal N (i)$ and $\mathcal E_s$ denote the set of all adjacent nodes to $i$ and the set of all links in the abstraction, respectively. Therefore, one can think of design parameter $d$ as an upper bound on the desired average number of neighbors of nodes in the abstracted network which is independent of the network size.}
For Property (ii), inequality \eqref{rho-approx} indicates that the resulting abstracted network $\NN(L_{s})$ has guaranteed performance bounds with respect to $\NN(L)$. The design constant  $\epsilon$ is referred to as  permissible performance loss parameter.  }


 \begin{table*}[t]
\begin{center}
    \begin{tabular}{ | p{7.0cm}  | l | p{8.0cm}  | }
    \hline
Homogeneous Systemic   Performance Measure & Symbol & Representation \\ \hline \hline
    Spectral Riemann zeta function & ${\zeta}_{q}(L)$ & $\displaystyle \Big( \sum_{i=2}^n \lambda_i^{-q} \Big)^{1/q}$  \\
    \hline
        Gamma entropy & $I_{\gamma}(L)$ & $\displaystyle  \gamma^2 \sum_{i=2}^n  \Big(\lambda_i- \big(\lambda_i^2-\gamma^{-2}\big)^{\frac{1}{2}} \Big)$ 
        \\
    \hline    
System Hankel norm        & $\eta(L)$ & $\displaystyle  \frac{1}{2}\lambda_2^{-1}
$  \\
    \hline  
Hardy-Schatten or $\mathcal{H}_p$ system norm              & $\theta_p(L)$ & $\displaystyle \left\{ \frac{1}{2\pi} \int_{-\infty}^{\infty} \sum_{k=1}^n \sigma_k(G(j \omega))^p  \hspace{0.05cm} d\omega \right\}^{1/p}$  $ = \alpha_0 \left( \tr \left( L^\dag \right)^{p-1}\right)^{\frac{1}{p}}$ 
        \\
    \hline  
Local Deviation Error for First Order Consensus Networks        & $\Delta(L)$ & $\frac{1}{2}\sum_{i=1}^n {d_i}^{-1}$  \\
    \hline
Local Deviation Error for Second Order Consensus Networks with $\beta >0$        & $\Upsilon(L)$ & $\frac{1}{2\beta}\sum_{i=1}^n {d_i}^{-2}$  \\
    \hline  
$\mathcal H_2$-norm of  Second Order Consensus Networks with $\beta >0$        & $\Theta_2(L)$ & $\left(\frac{1}{2\beta}\sum_{i=1}^n {\lambda_i}^{-2}\right)^{1/2}$  \\   
        \hline  
    \end{tabular}
        \caption{ \small Some important examples of homogeneous systemic performance measures.} \label{tab} 
\end{center}
   \vspace{-0.5cm}
\end{table*}

\section{Examples of Relevant Homogeneous Systemic Performance Measures}
	
We now present some existing and widely-used systemic performance measures for linear consensus networks; a list of these measures are  summarized in Table \ref{tab}.
	
\subsection{Sum of Homogeneous Spectral Functions}
\label{f_sum}

This class of performance measures is generated by forming summation of a given function of Laplacian eigenvalues. For a given matrix $L \in \mathfrak L_n$, suppose that  {$\varphi: \R_{+} \rightarrow \R_{+}$} is a decreasing homogeneous convex function. Then, the following spectral function 
\begin{equation}
	\rhoo (L) ~=~\sum_{i=2}^n \varphi(\lambda_i)
	\label{eq:512}
\end{equation}
is a homogeneous systemic measure \cite{SiamiSystemic}. Moreover, if {$\varphi$} is a homogeneous function of order $-\alpha$ where $\alpha>1$,  then its corresponding normalized index
\begin{equation}
	\Pi_\rhoo(L) ~=~ \left ( \sum_{i=2}^n \varphi(\lambda_i) \right)^{\frac{1}{\alpha}} 
	\label{eq:measure-normalized}
\end{equation}
is also a homogeneous systemic performance measure \cite{SiamiSystemic}. Some notable examples of this class of measures are discussed in the following parts. 

\subsubsection{Spectral Riemann Zeta Measures}
For a given network \eqref{first-order-G}, its corresponding spectral Riemann zeta function of order $q \geq 1$ is defined by
\begin{equation}
	{\zeta}_{q}(L)~:=~\bigg( \sum_{i=2}^n \lambda_i^{-q} \bigg)^{1/q},
	\label{zeta-measure}
\end{equation}
where $\lambda_2, \ldots,\lambda_n$ are eigenvalues of $L$ \cite{Hawking}. According to Assumption \ref{assum-coupling-graph}, all Laplacian eigenvalues  are strictly positive and, as a result, function \eqref{zeta-measure} is well-defined. According to the result presented in Subsection \ref{f_sum}, since $\varphi(\lambda)=\lambda^{-q}$ for $q \geq 1$ is a decreasing homogeneous convex function, the spectral function \eqref{zeta-measure} is a homogeneous systemic performance measure. The homogeneous systemic performance measure $\frac{1}{2}{\zeta}_{1}(L)$ is equal to the $\HH_2$-norm squared of a first-order consensus network \eqref{first-order-G} and ${\frac{1}{\sqrt{2}}\zeta}_{2}(L)$ equal to the $\HH_2$-norm of a second-order consensus model of a network of multiple agents (c.f. \cite{Siami13cdc}).


\subsubsection{Gamma Entropy}

The notion of gamma entropy arises in various applications such as {the design of minimum entropy controllers and interior point polynomial-time methods
in convex programming} with matrix norm constraints \cite{Blondel1999}. As it is shown in \cite{boyd97}, the notion of gamma entropy can be interpreted as a performance measure for linear time-invariant systems with random feedback controllers  by relating the gamma entropy to the mean-square value of the closed-loop gain of the system. The $\gamma$-entropy of network \eqref{first-order-G} is defined as
\begin{eqnarray*}
{ I_{\gamma}(L) := \begin{cases}
	\frac{-\gamma^2}{2\pi}\int_{-\infty}^{\infty} \log \det \big(I- \gamma^{-2} G(j\omega)G^*(j\omega) \big)d\omega~\\
	~~~~~~~~~~~~~~~~~~~~~~~~~~~~~~~~~~~~~~~~~~~~~~\text{for}~ \gamma \geq \|G\|_{\mathcal H_\infty}  \\
	\\
	\infty~~~~~~~~~~~~~~~~~~~~~~~~~~~~~~~~~~~~~~~~~~~~~~~~~~~\text{otherwise}
	\end{cases}}
	\label{gamma-formula}
\end{eqnarray*}
where $G(j\omega)$ is the transfer matrix of network \eqref{first-order-G} from $\xi$ to $y$ \cite{boyd97}. In \cite{SiamiSystemic}, it is shown that  the value of the $\gamma$-entropy for a given linear consensus network \eqref{first-order-G} can be explicitly computed in terms of Laplacian spectrum as follows
\begin{eqnarray}{
	I_{\gamma}(L)=\begin{cases}
	\displaystyle \sum_{i=2}^n ~f_{\gamma}(\lambda_i)~~~~~\gamma \geq \lambda_2^{-1}  \\
	\\
	\infty~~~~~~~~~~~~~~~\text{otherwise}
	\end{cases}},
	\label{gamma-formula}
\end{eqnarray}
where $f_{\gamma}(\lambda_i)=\gamma^2 \left( \lambda_i-\left(\lambda_i^2-\gamma^{-2}\right)^{\frac{1}{2}} \right)$.
Furthermore, the $\gamma$-entropy $I_{\gamma}(L)$ is a homogeneous systemic performance measure.

\subsection{Uncertainty volume}
	The uncertainty volume of the steady-state output covariance matrix of network \eqref{first-order-G} is defined by   
\begin{equation}
	|\Sigma|:= \det \Big( Y_{\infty} +{\frac{1}{n}J_n}\Big)
	\label{error ellipsoid}
\end{equation}
in which
\[ Y_{\infty} = \lim_{t \rightarrow \infty} \mathbb{E} \big[ y(t) y^{\text T}(t) \big].  \]
	This quantity is widely used as an indicator of the network performance \cite{siami14acc} and \cite{Mesbahi15cdc}. Since $y(t)$ is the error vector that shows distance from consensus, the quantity (\ref{error ellipsoid}) can be interpreted as volume of the steady-state error ellipsoid. It is straightforward to show this measure satisfies all properties of Definition \ref{def-homo-systemic}.

{\subsection{Hankel Norm} 

The Hankel norm of network \eqref{first-order-G}  and transfer matrix $G(j \omega)$ from $\xi$ to $y$ is defined as the $\mathcal L_2$-gain from past inputs to the future outputs, \ie
\[\|G\|_{H}^2~:=~\sup_{\xi \in L_2 (-\infty, 0] } \frac{\int _0^{\infty} y^{\rm T}(t) y(t) dt}{\int_{-\infty}^0 \xi^{\rm T}(t) \xi(t) dt}. \]
The value of the Hankel norm of network \eqref{first-order-G} can be equivalently computed using the Hankel norm of its disagreement form \cite{Olfati} that is given by 
\begin{eqnarray}
	\dot x_d(t)&=&-L_{d}  \, x_d(t) + M_n \, \xi(t), \label{first-order-d}\\
     	y(t)&=&M_n \hspace{0.05cm}x_d(t), \label{first-order-G-d}
\end{eqnarray}
where the disagreement vector is defined by 
\begin{equation}
	x_d(t) ~:=~ M_n \, x(t)~=~x(t) - \frac{1}{n}J_n \, x(t).
	\label{dis-vector-1}
\end{equation}
The disagreement network \eqref{first-order-d}-\eqref{first-order-G-d} is stable as the real part of every eigenvalue of the state matrix $-L_{d}=-(L+\frac{1}{n}J_n)$ is strictly negative. One can  verify that the transfer matrices from $\xi(t)$ to $y(t)$ in both realizations are identical. Therefore, the Hankel norm of the system from $\xi(t)$ to $y(t)$ in both representations are well-defined and equal, and is given by  \cite{hankel}
\begin{equation}
	\eta(L):=\|G\|_H= \sqrt{\lambda_{\max}(PQ)},
	\label{hankel}
\end{equation}
where the controllability Gramian $P$ is  the unique solution of
\begin{equation*}
	\Big(L + \frac{1}{n}J_n\Big)P+P\Big(L+\frac{1}{n}J_n \Big)- M_n =  0 
\end{equation*}
and the observability Gramian $Q$ is   the unique solution of
\begin{equation*} 
	Q\Big(L + \frac{1}{n}J_n\Big)+\Big(L+\frac{1}{n}J_n \Big)Q - M_n =  0. 
\end{equation*}
It is shown in \cite{SiamiSystemic} that the value of the Hankel norm of network \eqref{first-order-G} is equal to 
\[  \eta(L) ~=~ \frac{1}{2}\lambda_2^{-1}.\] 
One can verify that this measure is  a homogeneous systemic performance measure. 

\begin{remark}
One may also consider the sum of the $k$ largest eigenvalues of $L^{\dagger}$ as a performance measure. This is equivalent to evaluate the $k$ slowest modes of the network, which are the most energetic modes. This measure satisfies properties of Definition \ref{def-homo-systemic}  as it is convex and symmetric with respect to Laplacian eigenvalues  (c.f. \cite[Ch. 5.2]{borwein} and \cite{boyd2006}). 
\end{remark}

}

\subsection{Hardy-Schatten or $\mathcal{H}_p$ System Norms}
The $\mathcal H_p$-norm of networks (\ref{first-order-G})  for $2 \leq p \leq \infty $ is defined by 
	\begin{equation}
	\|G\|_{\mathcal H_p}~:=~ \left( \frac{1}{2\pi} \int_{-\infty}^{\infty} \sum_{k=1}^n \sigma_k(G(j \omega))^p  \hspace{0.05cm} d\omega \right )^{\frac{1}{p}},
	\label{h_p}
	\end{equation}
where $G$ is the transfer matrix  from $\xi(t)$ to $y(t)$ and  $\sigma_k(j \omega)$ for $k =1,\ldots,n$ are singular values of $G(j \omega)$. 
To ensure well-definedness of performance measure (\ref{h_p}), the marginally stable mode of the network must be unobservable through the  output. Thus, this performance measure remains well-defined as long as the coupling graph of the network stays connected. 
This class of system norms captures several important performance and robustness features of linear control systems. For instance, a direct calculation reveals that the  $\mathcal H_{2}$-norm of network  (\ref{first-order-G}) is 	
\begin{equation}
		\|G\|_{\mathcal H_2}~=~\left(\frac{1}{2}\sum_{i=2}^n\lambda_i^{-1}\right)^{\frac{1}{2}}.
	\end{equation}
This system norm quantifies the quality of noise propagation  throughout the network \cite{Siami14arxiv}. The $\mathcal H_{\infty}$-norm of a network is an  input-output system norm and its value for network (\ref{first-order-G}) is
	\begin{equation}
		\|G\|_{\mathcal H_{\infty}} ~=~ \lambda_2^{-1},
	\end{equation}
where $\lambda_2$ is known as the algebraic connectivity of the network \cite{Olfati}.  The value of $\mathcal H_{\infty}$-norm of network  (\ref{first-order-G}) can be interpreted as the worst attainable performance for all square integrable disturbance inputs. 

In \cite{SiamiSystemic}, the authors prove that the $\mathcal H_p$-norm of a  given network $\NN(L)$ is given by
		\begin{equation}
			\theta_p (L) ~:=~ \|G\|_{\mathcal H_p} = \alpha_0 \Big( \zeta_{p-1}(L) \Big)^{1-\frac{1}{p}}
		\label{sys-meas-H-p}
	\end{equation}
in which $\alpha_0^{-1}=\sqrt[p]{-\beta(\frac{p}{2},-\frac{1}{2})}$ and ${\beta}:\mathbb{R}\times \mathbb{R} \rightarrow \mathbb{R}$ is the well-known Beta function\footnote{$ \beta(x,y)=\int _{0}^{1}t^{x-1}(1-t)^{y-1}\,\mathrm {d} t$
for $\Rel\{x\}, \Rel \{y\} > 0$.}. Moreover, this measure is a homogeneous systemic performance measure for all $2 \leq p \leq \infty$.

\subsection{Local deviation error:} 
In network \eqref{first-order-G}, the local deviation of subsystem $i$ is equal to the deviation of the state of subsystem $i$ from the weighted average of states of its immediate neighbors, which can be formally defined by
	\begin{equation}
		\varepsilon_i(t)~:=~x_i(t)-\frac{1}{d_i} \sum_{e=\{i,j\} \in \mathcal E} w(e) \hspace{0.05cm} x_j(t).
		\label{eq:463}
	\end{equation}
The expected cumulative local deviation is then defined by
\begin{equation}
		\Delta(L) ~=~ \lim_{t \rightarrow \infty} \E \left [\sum_{i=1}^n \varepsilon_i(t)^2 \right]
	\label{global-measure}
	\end{equation}
with respect to input $\xi$ being a white noise process with identity covariance.  The notion of local deviation  can be extended and defined for velocity variables in the second-order consensus network \eqref{formation-1}-\eqref{formation-2} (c.f.,  \cite{Siami13cdc}) as follows
\begin{equation}
		\varsigma_i(t)~:=~v_i(t)-\frac{1}{d_i} \sum_{e=\{i,j\} \in \mathcal E} w(e) \hspace{0.05cm} v_j(t)
		\label{eq:463-2}
	\end{equation}
that is equal to the deviation of the velocity of subsystem $i$ from the weighted average of velocities of its neighbors. The expected cumulative local deviation is then given by
\begin{equation}
		\Upsilon(L) ~=~ \lim_{t \rightarrow \infty} \E \left [\sum_{i=1}^n \varsigma_i(t)^2 \right],
	\label{global-measure-2}
	\end{equation}
where it is assumed that input $\xi$ in network model \eqref{formation-1}-\eqref{formation-2} is a white noise process with identity covariance.

	\begin{theorem}
	\label{2-thm}
		The operators $\Delta, \Upsilon: \LL \rightarrow \R_+$ defined by (\ref{global-measure}) and \eqref{global-measure-2} are homogeneous systemic performance measures. Moreover, they can also be characterized as \begin{equation}
		\Delta(L) = \frac{1}{2} \sum_{i=1}^n {d_i}^{-1}
	\label{hahaha}
	\end{equation}
and
\begin{equation}
		\Upsilon(L) = \frac{1}{2 \beta} \sum_{i=1}^n {d_i}^{-2}
	\label{hahaha-2}
	\end{equation}
in which  $d_i$ is the degree of node $i \in \mathcal V$.
	\end{theorem}
	
\begin{proof}
	Let us define the total local deviation at time $t$ by
\begin{eqnarray}
\varepsilon_{\rm total}(t)~:=~\sum_{i \in \mathcal V} \varepsilon_i(t)^2.
\label{tot-loc}
\end{eqnarray}
We reformulate \eqref{eq:463} as
	\begin{eqnarray}
		\varepsilon_i(t)&=&d_i^{-1} \left(d_i x_i(t)- \sum_{e=\{i,j\} \in \mathcal E} w(e) \hspace{0.05cm} x_j(t)\right) \nonumber \\
		&=& d_i^{-1}  \sum_{e=\{i,j\} \in \mathcal E} w(e) \left( x_i(t) -  x_j(t)\right).
	\end{eqnarray}
Therefore, we get
\[ \varepsilon(t) = \diag \big [d_1^{-1}, \cdots, d_n^{-1} \big ] L \, x(t).\]
where $\varepsilon(t)$ is concatenation of elements $\varepsilon_i(t)$ for $i=1,\ldots,n$.	
Also, we can rewrite (\ref{tot-loc}) as follows
\[
\varepsilon_{\rm total}(t)~=~\varepsilon^{\rm T}(t) \varepsilon(t) ~=~ x^{\rm T}(t) Q x(t)
\]
with
\[Q~=~ L \, \diag \left [d_1^{-2}, \cdots, d_n^{-2} \right ] \, L.\]
Thus, according to  \cite[Thm. 5]{Siami14arxiv} the steady-state of $\varepsilon_{\rm total}$ is given by
\begin{eqnarray}
\Delta(L)&=&\lim_{t \rightarrow \infty}\E \left [\varepsilon_{\rm total}(t) \right]=\frac{1}{2}\tr \big( L^{\dag}Q\big)\nonumber \\
&=&\frac{1}{2}\sum_{i \in \mathcal V} {d_i}^{-1}.
\label{loc}
\end{eqnarray}
Now we show this measure is a homogeneous systemic performance measure. We first show that \eqref{loc} has property 1, which means 
\[ \Delta(\kappa L)~=~\frac{1}{2}\sum_{i \in \mathcal V} {(\kappa d_i)}^{-1}~=~\kappa^{-1} \Delta(L).\]
Furthermore, it is monotone, because if $L_1 \preceq L_2$ then we have 
\[ e_i^{\rm T} L_1 e_i \leq e_i^{\rm T} L_2 e_i,\]
where $e_i$ for $i=1,\ldots,n$ are the standard  basis for the $n$-dimensional Euclidean space.  Therefore, we have 
$L_1(i,i) \leq L_2(i,i)$ that guarantees the monotonicity of $\Delta$. Moreover, its convexity follows from convexity of function $1/x$ for all $x \in \R_+$.  Because consider two Laplacian matrices $L_1$
and $L_2$ with node degrees $d^{(1)}_i$ and $d^{(2)}_i$, respectively,  for $i=1,\ldots,n$. Then,  we get
\begin{eqnarray}
\Delta(cL_1 + (1-c)L_2)&=& \sum_{i \in \mathcal V} \frac{1}{c\,d_i^{(1)}+(1-c)\,d_i^{(2)}} \nonumber \\
 &\leq & \sum_{i \in \mathcal V} \left (\frac{c}{d_i^{(1)}}+\frac{1-c}{d_i^{(2)}} \right) \nonumber \\
  &= & \sum_{i \in \mathcal V} \frac{c} {d_i^{(1)}}+\sum_{i \in \mathcal V} \frac{1-c}{d_i^{(2)}} \nonumber \\
    &= & {c} \Delta(L_1)+(1-c) \Delta(L_2) \nonumber 
 \end{eqnarray}
for all $0\leq c \leq 1$.  This completes the proof of the first part. For the second part, let us define the total local deviation error at time $t$ by
\begin{equation}
\varsigma_{\rm total}(t)~:=~\sum_{i \in \mathcal V} \varsigma_i(t)^2.
\label{tot-loc-2}
\end{equation}
We similarly reformulate \eqref{eq:463-2} as
	\begin{eqnarray*}
		\varsigma_i(t)&=&d_i^{-1} \left(d_i v_i(t)- \sum_{e=\{i,j\} \in \mathcal E} w(e) \hspace{0.05cm} v_j(t)\right) \nonumber \\
		&=& d_i^{-1}  \sum_{e=\{i,j\} \in \mathcal E} w(e) \left( v_i(t) -  v_j(t)\right).
	\end{eqnarray*}
Therefore, we have
\[ \varsigma(t) = \diag \big [d_1^{-1}, \cdots, d_n^{-1} \big ] L \, v(t),\]
where $\varsigma(t)$ is concatenation of elements $\varsigma_i(t)$ for all $1\leq i \leq n$.	
Moreover, we can rewrite (\ref{tot-loc-2}) as follows
\begin{eqnarray}
\varsigma_{\rm total}(t)&=&\varsigma^{\rm T}(t) \varsigma(t) \nonumber \\
&=& v^{\rm T}(t) Q v(t),
\end{eqnarray}
where $Q$ is given by $Q~=~ L \, \diag \left [d_1^{-2}, \cdots, d_n^{-2} \right ] \, L$.
Therefore, the steady-state of $\varsigma_{\rm total}$ can be characterized as
\begin{eqnarray}
\Upsilon(L)&=&\lim_{t \rightarrow \infty}\E \left [\varsigma_{\rm total}(t) \right]=\frac{1}{2\beta }\tr \big( (L^{\dag})^2Q\big)\nonumber \\
&=&\frac{1}{2\beta }\sum_{i \in \mathcal V} {d_i}^{-2}.
\label{loc-2}
\end{eqnarray}
This measure is a homogeneous systemic performance measure. It is straightforward to  show that \eqref{loc-2} satisfies property 1 by verifying that 
\[ \Upsilon(\kappa L)~=~\frac{1}{2\beta }\sum_{i \in \mathcal V} {(\kappa d_i)}^{-2}~=~\kappa^{-2} \Upsilon(L).\]
It is monotone, as if $L_1 \preceq L_2$, then we have 
\[ e_i^{\rm T} L_1 e_i \leq e_i^{\rm T} L_2 e_i.\]
As a result, it follows that
$L_1(i,i) \leq L_2(i,i)$ that guarantees the monotonicity of $\Upsilon$. Finally, its convexity can be concluded from convexity of function $1/x^2$ for all  $x \in \R_+$. 
	\end{proof}

\begin{remark}
For first-order consensus network (\ref{first-order-G}) that are defined over $d$-regular coupling graphs, the corresponding microscopic measure (\ref{hahaha}) scales linearly with network size. For regular lattices that are $d$-regular graphs, our result assumes the reported result of \cite{Bamieh12} as its special case.
\end{remark}

\begin{figure}[t]
\centering
\begin{tikzpicture}[scale=1.4]
\draw [fill] (-4,0) circle [radius=0.1];
\draw [fill] (-3.5,0) circle [radius=0.1];
\draw [fill] (-3,0) circle [radius=0.1];
\draw [fill] (-4,.5) circle [radius=0.1];
\draw [fill] (-3.5,0.5) circle [radius=0.1];
\draw [fill] (-3,0.5) circle [radius=0.1];
\draw [ thick] (-4,0) -- (-4,0.5);
\draw [ thick] (-3.5,0) -- (-3.5,0.5);
\draw [ thick] (-3,0) -- (-3,0.5);
\draw [ thick] (-3.5,0) -- (-3,.5);
\draw [ thick] (-4,0) -- (-3,0);
\draw [ thick] (-4,0.5) -- (-3.5,0.5);
\draw [fill] (0,0) circle [radius=0.1];
\draw [fill] (.5,0) circle [radius=0.1];
\draw [fill] (.25,0.25 ) circle [radius=0.1];
\draw [fill] (1,0) circle [radius=0.1];
\draw [fill] (0,.5) circle [radius=0.1];
\draw [fill] (.5,0.5) circle [radius=0.1];
\draw [ thick] (0,0) -- (0,0.5);
\draw [ thick] (0,0) -- (0.5,0.5);
\draw [ thick] (.5,0) -- (1,0);
\draw [ thick] (.5,0) -- (.5,.5);
\draw [ thick] (0,0) -- (0.5,0);
\draw [ thick] (0,.5) -- (0.5,.5);
\node[] at (-3.5,-1) {$(a)$};
\node[] at (.5,-1) {$(b)$};
\end{tikzpicture}
  	\caption{{Two isospectral graphs with six nodes \cite{newman2001laplacian}.}}
  	\label{figiso}
\end{figure}
\begin{remark}
Fig. \ref{figiso} shows example of two isospectral\footnote{Two graphs are called isospectral if and only if  their Laplacian matrices have the same multi-sets of eigenvalues} graphs that are not isometric\footnote{This means that their adjacency matrices are not permutation-similar.}.  While the value of a spectral systemic performance measure is equal for both graphs, the value of an expected cumulative local deviation measure is different for each of these graphs and depend on their specific interconnection topology. This simple observation implies that systemic performance measures (\ref{hahaha}) and (\ref{hahaha-2}) are suitable tools to differentiate among networks with isospectral coupling graphs. 
\end{remark}

\section{Abstraction with Guaranteed Bounds}
\label{spars}
In this section, we develop a fast abstraction algorithm for the class of linear consensus networks \eqref{first-order-G} with guaranteed bounds with respect to the class of homogeneous systemic performance measures.  

\subsection{Intrinsic Tradeoffs on the Best Achievable  Abstractions}
The abstraction goals are to reduce the number of feedback links while preserving a desired level of performance.  From notation \eqref{local-measure}, one can easily verify that the value of $\mathcal{S}_{0,1}$-measure is equal to the maximum of $|\mathcal{N}(i)|$ for all nodes $i=1,\ldots,n$, which makes it a suitable surrogate for design parameter $d$. The next result reveals an inherent interplay between sparsity and performance. 

\begin{theorem}\label{fund-trade-thm}
For a given network \eqref{first-order-G} that is endowed with a homogenous systemic performance measure $\rho: \LL \rightarrow \R_+$ of order $-\alpha$, suppose that $w_*= \max_{e \in \mathcal E} w(e)$. Then, there are fundamental tradeoffs between normalized performance and graph sparsity measures in the following sense 
\begin{equation}
\Pi_{\rho}(L)~\|A\|_{\ell_0}~\geq~ 2\varrho^*  (n-1), \label{tradeoff-1}
\end{equation}
and 
\begin{equation}
\Pi_{\rho}(L)~\|A\|_{\mathcal{S}_{0,1}}\geq~2\varrho^*\label{tradeoff-2}
\end{equation}
when $n > 2$, in which $A$ is the adjacency matrix of the coupling graph and $\varrho^*=w_*^{-1}\Pi_{\rho}(L_{\mathcal{K}_n})$, where $L_{\mathcal{K}_n}$ is Laplacian matrix of the unweighted complete graph. 
\end{theorem}
\begin{proof}
Since it is assumed that the coupling graph of the network is connected, the sparsity measure $\|A\|_{\ell_0}$  is always bounded from below by $2(n-1)$; with equality sign if the coupling graph is a {tree}. Thus, the following inequality holds on the cone of positive semidefinite matrices 
\[ L ~\preceq~   \frac{\|A\|_{\ell_0}}{2(n-1)} \hspace{0.05cm} L ~\preceq~ \frac{w_* \|A\|_{\ell_0}}{2(n-1)} \hspace{0.05cm}L_{\mathcal{K}_n}.  \]
From monotonicity property, it follows that 
\[ \rho(L) \hspace{0.05cm} \|A\|_{\ell_0}^{\alpha} ~\geq~ {2^{\alpha}}(n-1)^{\alpha} w_*^{-\alpha}\hspace{0.05cm}  \rho(L_{\mathcal{K}_n}). \]
By taking $\alpha$'th root from both sides, one can conclude the desired inequality \eqref{tradeoff-1}. When $n > 2$, the localized sparsity measure $\|A\|_{\mathcal{S}_{0,1}}$ is always greater of equal to $2$. Therefore, the following relation holds
\[ L ~\preceq~   \frac{1}{2} \hspace{0.05cm} \|A\|_{\mathcal{S}_{0,1}} \hspace{0.05cm} L~\preceq~   \frac{1}{2} \hspace{0.05cm} w_* \|A\|_{\mathcal{S}_{0,1}} \hspace{0.05cm}L_{\mathcal{K}_n}.  \]
By utilizing the monotonicity property, we get
\[ \rho(L) \hspace{0.05cm} \|A\|_{\mathcal{S}_{0,1}}^{\alpha}  ~\geq~2^{\alpha} w_*^{-\alpha}\hspace{0.05cm}  \rho(L_{\mathcal{K}_n}).\]
The desired inequality \eqref{tradeoff-2} follows from taking $\alpha$'th root from both sides of the inequality.  
\end{proof}	

The monotonicity property of a  systemic performance measure implies that link removal will lead to performance deterioration. Theorem \ref{fund-trade-thm} quantifies this inherent interplay by saying that sparsity and performance cannot be improved indefinitely both at the same time. As we will see in the following subsection, this is exactly why we need to perform reweighing after link elimination procedure in order to achieve an approximation that meets \eqref{rho-approx}.


\subsection{Existence and Algorithms}
\label{sec:existence}

The next theorem enables us to harness the monotonicity property of homogeneous systemic measures in our network approximations.

\begin{theorem}
\label{th-abstract}
Suppose that two linear consensus networks $\NN(L)$ and $\NN(L_{s})$ are endowed with a homogeneous systemic performance measure $\rhoo:\LL \rightarrow \R_+$ of order $-\alpha$. For a given constant $\epsilon \in (0,1)$, the two networks are $\epsilon$-approximation of each other, i.e., property \eqref{rho-approx} holds, if and only if their state matrices satisfy 
		\begin{equation}
			(1 - \epsilon) L ~\preceq~ L_s ~\preceq~(1 + \epsilon) L. \label{proximity-cone}
		\end{equation}
\end{theorem}
\begin{proof}
According to the monotonicity and homogeneity properties of system measures, it follows that if (\ref{proximity-cone}) holds then we have 
\begin{equation} 
(1+ \epsilon)^{-\alpha} \rhoo(L) ~\leq~ \rhoo(L_s) ~\leq~ (1-\epsilon)^{-\alpha}\rhoo(L). 
\label{eq:432}
\end{equation}
Therefore, according to \eqref{eq:432} and \eqref{rho-approx}, $\NN(L_{s})$ is an $\epsilon$-approximation of $\NN(L)$. Let us consider the following measures
\begin{equation}
\rhoo_{v}(L) \, = \,  v^{\rm T} L^{\dag} v
\end{equation}
for all $v \in \R^n$. This operator is a homogeneous systemic performance measure of order $-1$. For all $v \notin \text{Span}\{\mathbbm 1\}$, inequality   \eqref{rho-approx} yields 
\[ -\epsilon \,\leq \, \frac{\rho_{v}(L)-\rho_{v}(L_s)}{\rho_{v}(L_s)} \, \leq \, \epsilon. \]
Thus, it follows that
\begin{equation}
(1+\epsilon)^{-1} \, \leq \, \frac{v^{\rm T}L_s^{\dag}v}{v^{\rm T}L^{\dag}v} \, \leq \, (1-\epsilon)^{-1}. \label{eq:598}
\end{equation}
Since $v^{\rm T}L^{\dag}v>0$, inequalities \eqref{eq:598} can be rewritten as \begin{equation}
  (1+\epsilon)^{-1}{v^{\rm T}L^{\dag}v} \, \leq \, {v^{\rm T}L_s^{\dag}v} \, \leq \, (1-\epsilon)^{-1}{v^{\rm T}L^{\dag}v}. 
 \label{eq:603}
 \end{equation}
We know that $L$ and $L_s$ are Laplacian matrices and \eqref{eq:603} holds for all $v \notin \text{Span}\{\mathbbm 1\}$; therefore, we get
\[ (1+\epsilon)^{-1} L^{\dag} \, ~\preceq~ \, L_s^{\dag} \, ~\preceq~ \,  (1-\epsilon)^{-1}L^{\dag}. \] 
This inequality can be rewritten to obtain the desired result
\[ (1-\epsilon)L \, ~\preceq~ \, L_s \, ~\preceq~ \, (1+\epsilon) L. \] 
\end{proof}	

The result of the above theorem is crucial as it enables us to take advantage of monotonicity property of systemic performance measures in our approximations. For two given networks $\NN(L_1)$ and $\NN(L_2)$, inequality $\rho(L_2) \leq  \rho(L_1)$ can be realized through several possible scenarios; for example, network $\NN(L_2)$ can be constructed by (i) adding new weighted edges to the coupling graph of network $\NN(L_1)$, (ii)    increasing weights of some of the existing links in network $\NN(L_1)$, (iii)  rewiring topology of network $\NN(L_1)$ while ensuring $L_1 \preceq  L_2$.

The next result proves existence of an abstraction for every given linear consensus network.

\begin{theorem}\label{th-main}
Suppose that a network $\NN(L)$ with coupling graph $\G=(\V, \EE,w)$ endowed with a homogeneous systemic performance measure $\rhoo:\LL \rightarrow \R_+$ of order $-\alpha$ and a design parameter $d >2$ are given. Then, it is possible to construct another network $\NN(L_{s})$ with coupling graph $\G_s=(\V, \EE_s,w_s)$ such that 

\noindent (i)  $\NN(L_s)$ is a $({ \frac{\sqrt{8d}}{d+2}}, d)$-abstraction of network $\NN(L)$;\\
\noindent (ii) $\EE_s \subset \EE$.
\end{theorem}
\begin{proof}
{ It is well known that the characteristic polynomial of $A + v v^{\rm T}$ can be computed based on the characteristic polynomial\footnote{ The characteristic polynomial of matrix $A \in \R^{n \times n}$ is defined by \[p_A(t) = \det \left(tI_n - A\right).\]} of $A$, its eigenvalues $\mu_i$, and eigenvectors $u_i$ as follows
\[ P_{A+ v v^{\rm T}}(x) = P_A(x) \left ( 1 - \sum_i \frac {\langle v, u_i \rangle^2}{x- \mu_i} \right) \]
in which $\langle v, u_i \rangle:= v^{\text T} u_i$.
 Moreover, we know that the eigenvalues of $A + v v^{\rm T}$ interlace those of $A$. Reference\cite{Spielman14} suggests a framework  which provides intuition as to where these new eigenvalues are located.  If we start with an empty graph on $n$ nodes, then after $M=\lceil {d(n-1)/2} \rceil$ iterations of choosing links from the original graph and adjusting its weight, it is shown that the eigenvalues of the resulted graph are controlled by maintaining two barrier potential functions. According to \cite[Th. 1.1]{Spielman14}, coupling graph $\G=(\V, \EE,w)$ has a weighted subgraph $\hat {\mathcal G}=(\V, \hat \EE, \hat w)$ with $|\hat \EE|=\lceil {d(n-1)}/2 \rceil$ that satisfies
\begin{equation}
L ~\preceq~ L_{\hat{\mathcal G}} ~\preceq~ \left(\frac{1+\sqrt{d/2}}{1-\sqrt{d/2}}\right)^2 L
\label{eq:497}
\end{equation}
where $L_{\hat{\mathcal G}}$ is the Laplacian matrix of graph $\hat{\mathcal G}$. We define $\mathcal G_s=(\V, \hat \EE, w_s)$ by its Laplacian matrix, which is given by
\begin{equation}
L_s ~:=~  \frac{(1-\sqrt{d/2})^2}{1+d/2} L_{\hat{\mathcal G}}.  
\label{eq:501}
\end{equation}
Therefore, according to \eqref{eq:497} and \eqref{eq:501}, it follows that 
\begin{equation}
\left(1-\frac{\sqrt{2d}}{d/2 +1}\right) L ~\preceq~ L_{s} ~\preceq~  \left(1+\frac{\sqrt{2d}}{d/2 + 1}\right) L.
 \label{eq:505}
\end{equation}
Using \eqref{eq:505} and Theorem \ref{th-abstract}, it yields that $\NN(L_{s})$ is a $(\frac{\sqrt{2d}}{d/2+1}, d)$-abstraction of $\NN(L)$. { Since $\mathcal E_s= \hat{\mathcal E}$ and the fact that $\hat{\mathcal E}$ is obtained by taking samples from $\mathcal E$, one concludes that $\mathcal E_s  \subset \mathcal E$.}}
\end{proof}	

In our next result, we show that every consensus network has a sparsification such that: (i) it yields a better systemic performance than the original network, and (ii) the total weight sum of the coupling graph of the spars network is controlled, i.e., it is less than a constant multiple, which is independent of the network size, of the total weight sum of the original network.

\begin{proposition}\label{coro-1}
For a given consensus network $\NN(L)$ with coupling graph $\G=(\V, \EE,w)$ and every $d > 2$, there exists a consensus network $\NN(L_s)$ with coupling graph $\G_s=(\V, \EE_s,w_s)$ that has at most $dn/2$ links and $\EE_s \subset \EE$. Moreover, we have  \\
\vspace{.1cm}
\noindent (i) the total weight of coupling graph of $\NN(L_{s})$ is controlled, i.e., 
\begin{equation}
			\sum_{e\in \EE_s} w_s(e) ~\leq~ {\left(\frac{\sqrt{2d}+2}{\sqrt{2d}-2}\right)^2} \sum_{e\in \EE} w(e);
		\label{same-order}
		\end{equation}
 \noindent (ii) $\NN(L_s)$ has a superior performance with respect to  $\NN(L)$, i.e., 
\[ \rhoo\left(L_{s}\right) ~\leq~ \rhoo\left(L\right),\]
		for every homogeneous systemic performance measure $\rhoo:\LL \rightarrow \R_+$.
\vspace{.1cm}
\end{proposition}
\begin{proof}
{According to Theorem \ref{th-main}, $\NN(L)$ has a $(\frac{\sqrt{8d}}{d+2}, d)$-abstraction, $\NN(L_{\hat s})$, with $\mathcal E_{\hat s} \subset  \EE$, which means that we have 
\begin{equation}
		\left (1-\frac{\sqrt{2d}}{d/2 + 1}\right )^{\alpha} ~\leq~ \frac{\rhoo\left(L\right)}{\rhoo\left(L_{\hat s}\right)} ~\leq~ \left (1+\frac{\sqrt{2d}}{d/2 + 1} \right)^\alpha
		\label{eq:533}
\end{equation}
for every homogeneous systemic performance measure $\rhoo:\LL \rightarrow \R_+$ of order $-\alpha$. Let us define $L_s=\frac{d+2-\sqrt{8d}}{d + 2} {L}_{\hat s}$. Then, it follows that
\begin{equation}
\rhoo({L}_{\hat s})~=~\left(\frac{d + 2}{d+2-\sqrt{8d}}\right)^{\alpha} \rhoo(L_{s}).
\label{eq:537}
\end{equation}
Finally, by substituting $\rhoo({L}_{\hat s})$ from \eqref{eq:537} in the left-hand side inequality of \eqref{eq:533},  we obtain
\[ \rhoo\left(L_{s}\right) ~\leq~ \rhoo\left(L\right). \]
It is straightforward to verify that operator $\Xi(L)=\left(\sum_{e\in \EE} w(e)\right)^{-1}$ is a  homogeneous systemic performance measure with $\alpha=1$. From the left-hand side inequality of \eqref{eq:533}, we get the desired bound \eqref{same-order}.}
\end{proof}

{ 
\begin{remark}
Proposition \ref{coro-1} demonstrates a tradeoff between the largest possible number of links in an abstraction, i.e., upper bound of $|\EE_s|$, and the best achievable ratio of the total weights, i.e., upper bound of $\sum_{e\in \EE_s} w_s(e)/ \sum_{e\in \EE} w(e)$. We have that 
\begin{equation}
|\EE_s|~\leq~ \frac{1}{2}dn
\label{eq:1157}
\end{equation} 
and  
\begin{equation}
\frac{\sum_{e\in \EE_s} w_s(e)}{ \sum_{e\in \EE} w(e)} ~\leq~ \left(\frac{\sqrt{2d}+2}{\sqrt{2d}-2}\right)^2.
\label{eq:1162}
\end{equation} 
Let us fix $n$. By increasing the value of parameter $d$, the upper bound in \eqref{eq:1157} increases, but the upper bound in \eqref{eq:1162} decreases and gets closer to $1$. On the other hand, if $d \rightarrow 2$, then the upper bound in \eqref{eq:1162} tends to infinity. 
\end{remark}
}
 
{ We next employ a randomized algorithm to compute an $(\epsilon,d)$-abstraction of a given network. A randomized algorithm utilizes a degree of randomness as part of its logic. Randomization allows us to design provably accurate algorithms for problems that are massive and computationally expensive or NP-hard. For this aim based on \cite{Spielman}, we sample low-connectivity coupling links with high probability and high-connectivity coupling links with low probability. For a given consensus network $\NN(L)$ with $n$ nodes, we sample links of the coupling graph of this network $M:=\lceil {dn}/2 \rceil$ times in order to produce an $(\epsilon, d)$-abstraction. Let us denote probability of selecting a link $e \in \mathcal E$ by  $\pi(e)$\footnote{ It is well-know that $\sum_{e \in E}w(e)r(e)=n-1$; therefore, we have  $\sum_{e \in \mathcal E} \pi(e) = 1$.} that is proportional to $w(e)r(e)$, where $w(e)$ and $r(e)$ are the weight and the effective resistance of link $e$, respectively. In each step of sampling\footnote{ A sampling is a discrete probability distribution on a support $\mathcal E$ of all possible samples. The probability of selecting link $e$ is denoted by positive number $\pi(e)$ for all $e \in \mathcal E$. We also have $\sum_{e \in \mathcal E} \pi(e) = 1$, because $\pi$ is a probability distribution on $\mathcal E$.}, we add the selected link $e$ to the abstraction with weight $w(e)/(M \pi(e))$.  All details of our proposed algorithm are explained below. The following result, which is obtained based on a theorem in \cite[Thm. 1]{Spielman}, provides us with a certificate that the above randomized algorithm is capable of generating a proper abstraction of a given linear consensus network.

{
\begin{algorithm}[t]
    \SetKwInOut{Input}{Input}
    \SetKwInOut{Output}{Output}

{
    \Input{$\mathcal G=(\mathcal V, \mathcal E, w)$ and $r(e)$ for all $e \in \mathcal E$}
            \vspace{.2cm}
    \Output{$\mathcal G_s=(\mathcal V, \mathcal E_s, w_s)$ }
    		\vspace{.4cm}
}
   
    {\bf set} $\mathcal G_s$ to be the empty graph on $\mathcal V$ (\ie ~$\mathcal E_s := \{\}$ and $w_s(.):=0$) \\
    {\bf set} {$\pi(e)= \frac{w(e)r(e)}{n-1}$ for all $e \in \mathcal E$}\\
        \For{$j = 1$ {\it to} $M:=\lceil {dn}/2 \rceil$}
       { sample a link from $\mathcal E$ with probability distribution $\pi$ $\rightarrow$ $\{e\}$\\
       $\mathcal E_s = \mathcal E_s \cup \{e\}$\\
       $w_s(e) = w_s(e) + \frac{w(e)}{M \pi(e)}$}
        \Return $\mathcal G_s$\\
      
    \caption{Network Abstraction Algorithm}
\end{algorithm}	}

\begin{theorem}
\label{random_al}
Suppose that a linear consensus network $\NN(L)$ endowed with a homogeneous systemic performance measure $\rhoo:\LL \rightarrow \R_+$ of order $-\alpha$ is given and  a  permissible performance loss parameter $\epsilon \in ({1}/{\sqrt{n}},1]$ is fixed. Let us pick a real number $d$ that is at the order of $\epsilon^{-2} \log n$. Then, Algorithm 1 produces an $(\epsilon,d)$-abstraction of network $\NN(L)$, whose coupling graph is a subgraph of $\NN(L)$,  with probability at least $0.5$.
 
	\end{theorem}

}

\begin{proof}
Let us consider the following projection matrix
	\begin{equation}
	 P ~=~ W^{1/2} E L^{\dag} E^{\rm T} W^{1/2},
	 \label{}
	 \end{equation}
where $E$ is $m$-by-$n$ incidence matrix and $W$ is a diagonal matrix with link weights on its diagonal such that $L=E^{\rm T} W E$. 
The $m$-by-$m$ matrix $P$ has  eigenvalue at $0$ with multiplicity  $m-n+1$  and eigenvalue at $1$ with multiplicity  $n-1$ \cite[Lemma 3]{Spielman}.
Now, we show that  the sampling of links in Algorithm 1 corresponds to selecting $M = \mathcal O(n \log n/ \epsilon^2)$ columns at random from matrix $P$. Then by a concentration lemma of Rudelson \cite[Thm. 3.1]{Rudelson:2007} and Markov's inequality, {with probability at least $0.5$}, we get 
\begin{equation}
 \| P -P \Gamma P \|_2 ~\leq~ \epsilon,
 \label{eq:110}
 \end{equation}
where $\Gamma$ is a nonnegative diagonal matrix. Then, it is straightforward to show that for every homogeneous systemic {performance} measure $\rho: \LL \rightarrow \R_+$, we have
\begin{equation*}
	\left| \frac{\Pi_{\rho}(L)-\Pi_{\rho}(L_s)}{\Pi_{\rho}(L_s)}	\right|\leq \epsilon		
	\end{equation*}  
in which $L_s= E W^{1/2} \Gamma W^{1/2} E^{\rm T}$.
One can show that the inequality \eqref{eq:110} is equivalent to
 \begin{equation}
\sup_{x \in \R^m \atop x \neq 0 } \frac{| x^{\rm T} (P -P \Gamma P) x |}{x^{\rm T} x}~ \leq~ \epsilon.
\label{eq:1125}
 \end{equation}
Since $\text{Im}\{W^{1/2}  E\} \subset \R^m$, it follows that 
 \begin{eqnarray*}
 \sup_{x\in \text{Im}\{W^{1/2}  E\} \atop x \neq 0 } \frac{| x^{\rm T} (P -P \Gamma P) x |}{x^{\rm T} x} & \leq & \sup_{x \in \R^m \atop x \neq 0 } \frac{| x^{\rm T} (P -P \Gamma P) x |}{x^{\rm T} x} \\
& \leq & \epsilon.
\end{eqnarray*}
Let us define $x= W^{1/2} E x'$. Then, we can rewrite \eqref{eq:1125} as follows
 \begin{equation}
\sup_{x' \in \R^n \atop x'\notin \text{ker}\{W^{1/2}  E\}} \frac{| x'^{\rm T} (L - L_s) x' |}{x'^{\rm T} L x'} ~\leq~ \epsilon.
\label{eq:?}
 \end{equation} 
For all $x'\in \text{ker}\{W^{1/2}  E\}$, one gets $x'^{\rm T} L x' = x'^{\rm T} L_s x'=0$. As a result, it follows that 
 \begin{equation}
\sup_{x' \in \R^n \atop x'\neq 0} \frac{| x'^{\rm T} (L - L_s) x' |}{x'^{\rm T} L x'} ~\leq~ \epsilon,
 \end{equation}
 which implies that 
 \begin{equation}
 (1-\epsilon) L  ~\preceq~ L_s = E W^{1/2} \Gamma W^{1/2} E^{\rm T}  ~\preceq~ (1+\epsilon) L.
 \end{equation}
 Finally, using this and Theorem \ref{th-abstract}, we conclude the desired result. 
	\end{proof}
	
Algorithm 1 produces a network abstraction with $\mathcal O(n \log n/\epsilon^2)$ feedback links  in expectation and runs in approximately linear time $\tilde {\mathcal O}(m)$, where $m$ is the number of links  (c.f. \cite{Spielman:2004}). This favorable almost-linear-time complexity is achieved by having access to good approximations of all effective resistances. In {\cite{Spielman}}, the authors show that $\mathcal O(\log n)$ calls to a solver for a linear system of equations with symmetric diagonally dominant (SDD) matrix can provide sufficiently good approximations to all effective resistances. Moreover, it is shown in \cite{Spielman14} that a  spectral sparsification with $\mathcal O(n/\epsilon^2)$ links can be computed in $\mathcal O(n^3m/\epsilon^2)$ time by employing a slower deterministic algorithm for link selection.
The best known classical algorithm for calculating effective resistances relies on solving a Laplacian linear system and takes $ \tilde {\mathcal O}(m)$ time\cite{Spielman14,Fung}. 

\begin{remark}
By putting together results of Theorems \ref{th-main} and \ref{random_al},  we observe  an intrinsic tradeoff between the number of feedback links $M$ and the permissible performance loss parameter $\epsilon$. These two design factors move in opposite directions, i.e.,  if we decrease $M$, $\epsilon$ increases, and vice versa.  According to Theorems \ref{th-main} and \ref{random_al}, one can deduce that  the number of feedback links $M$ decreases inversely with the square of $\epsilon$. 
\end{remark}

\subsection{Guaranteed Performance  Bounds}

In the following, we show that our proposed abstraction algorithm approximately preserves frequency characteristics of the original (dense) network  (see Fig. \ref{h_2_error_block}). Our abstraction method shares some common roots with the classical model reduction techniques, where the objective is to find a  reduced model that yields  small $\mathcal H_2$-norm error (c.f., \cite{H2error}). Our first result gives  a tight upper bound on the $\mathcal H_2$-norm error of two linear consensus networks in terms of their Laplacian matrices.

\begin{lemma}
	\label{lemma-spars-norm}
Suppose that $\NN(L)$ and $\NN(\hat{L})$ are two given consensus networks governed by dynamics \eqref{first-order-G}. Then, we have
			\begin{equation} \frac{\| G-\hat{G}\|^2_{\mathcal H_2}}{\| G\|^2_{\mathcal H_2}} ~\leq~ \frac{\tr \big(\hat{L}^{\dag}+L^{\dag}-4(L+\hat{L})^{\dag}\big)}{\tr \left (L^\dag \right)}, \label{eq:908}
			\end{equation}
where $G(s)$ and $\hat{G}(s)$ are transfer matrices of $\NN(L)$ and $\NN(\hat{L})$ from input $\xi$ to output $y$ , respectively. 
	\end{lemma}
\begin{proof}
In the first step, we define an augmented dynamical network $\NN^*$ using the two given networks 
\begin{eqnarray}
\begin{cases}
\dot{z}(t) \,=\, -\begin{bmatrix}
       L+\frac{1}{n}J_n &  0 \\[0.2em]
        0  & L_{s}+\frac{1}{n}J_n  \end{bmatrix}z(t)+\begin{bmatrix} M_n \\ M_n \end{bmatrix} \xi (t)\\ 
        \\
y(t)\,=\,\begin{bmatrix} M_n & -M_n \end{bmatrix} z(t)
\label{first-order-D}
\end{cases}
\end{eqnarray}
in which $z \in \R^{2n}$, $\xi \in \R^n$, $y \in \R^n$.  If we denote the transfer matrix of network \eqref{first-order-D} from $\xi$ to $y$ by $G^*$, then one can show that 
	\begin{equation}
		\| G^* \|_{\mathcal H_2}~=~\| G_{}-G_{s}\|_{\mathcal H_2}.
		\label{eq:924}
	\end{equation}
It is known that calculating the $\mathcal H_2$-norm of a linear time-invariant system reduces to 
solving an Algebraic Lyapunov Equation (ALE) \cite{Doyle89}. Let us form the corresponding ALE 
to \eqref{eq:924} using the state matrices of  the augmented network
\begin{equation}
	\mathfrak A X ~+~X \mathfrak A~=~\begin{bmatrix} M_n & -M_n\\ -M_n & M_n \end{bmatrix},
	\label{lyp_eq}
	\end{equation}
	where 
	\[\mathfrak A~=~-\begin{bmatrix}
       L_{}+\frac{1}{n}J_n &  0_{n \times n} \\[0.3em]
         0_{n \times n}  & L_{s}+\frac{1}{n}J_n  \end{bmatrix}\]
       and 
       \[X~=~\begin{bmatrix} X_1 & X_2\\ X_2^{\text T} & X_3\end{bmatrix}. \]       
The matrix equation (\ref{lyp_eq}) can be decomposed into three Sylvester equations as follow.  The first equation is 
	\begin{equation*}
	\left(L_{}+\frac{1}{n}J_n\right)X_1~+~X_1\left(L_{}+\frac{1}{n}J_n\right)~=~M_n
	\end{equation*}
and its solution is given by 
	\begin{equation}
	X_1~=~\frac{1}{2} L_{}^{\dag}.
	\label{eq:1057}
	\end{equation} 
	The second equation is 
	\begin{equation*}
	\left (L_{s}+\frac{1}{n}J_n \right )X_3~+~X_3\left (L_{s}+\frac{1}{n}J_n\right)~=~M_n
	\end{equation*}
and its unique solution is given by
	\begin{equation}
	X_3~=~\frac{1}{2} L_{s}^{\dag}.
	\label{eq:1067}
	\end{equation}
	 Finally, the third one is 
			\begin{equation}
	\left(L_{}+\frac{1}{n}J_n\right)X_2~+~X_2\left(L_{s}+\frac{1}{n}J_n\right)~=~-M_n
	\end{equation}
with unique solution
	\[ X_2~=~-\int_{0}^{\infty}e^{-(L_{}+\frac{1}{n}J_n)t}M_ne^{-(L_{s}+\frac{1}{n}J_n)t}dt,\]
where the integrand can be reformulated as 
\begin{eqnarray*}
e^{-(L+\frac{1}{n}J_n)t}M_ne^{-(L_{s}+\frac{1}{n}J_n)t} \hspace{-0.2cm} & =  & \hspace{-0.2cm} e^{-(L+\frac{1}{n}J_n)t}e^{-(L_{s}+\frac{1}{n}J_n)t} \nonumber \\
&&\hspace{-0.6cm} +~~ \frac{1}{n}e^{-(L+\frac{1}{n}J_n)t}J_ne^{-(L_{s}+\frac{1}{n}J_n)t}  \nonumber \\
&=&e^{-(L+\frac{1}{n}J_n)t}e^{-(L_{s}+\frac{1}{n}J_n)t}+\frac{e^{-2t}}{n}J_n.
	\end{eqnarray*}
By utilizing the Golden-Thompson inequality for Hermitian matrices, it follows that 
\begin{eqnarray}
	&&\hspace{-0.7cm}\tr \left(e^{-(L_{}+\frac{1}{n}J_n)t}M_ne^{-(L_{s}+\frac{1}{n}J_n)t} \right)  \nonumber \\
	&&\hspace{2cm}=~ \tr \left(e^{-(L_{}+\frac{1}{n}J_n)t}e^{-(L_{s}+\frac{1}{n}J_n)t}+\frac{e^{-2t}}{n}J_n \right) \nonumber \\
	&&\hspace{2cm} \geq ~ \tr \left(e^{-(L_{}+L_s+\frac{2}{n}J_n)t}+\frac{e^{-2t}}{n}J_n\right)
	\end{eqnarray}
Therefore, the trace of $X_2$ can be bounded by
	\begin{eqnarray}
	\tr(X_2)&=&- \tr \left( \int_{0}^{\infty}e^{-(L_{}+\frac{1}{n}J_n)t}M_ne^{-(L_{s}+\frac{1}{n}J_n)t}dt \right) \nonumber \\
	&=&-  \int_{0}^{\infty} \tr \left(e^{-(L_{}+\frac{1}{n}J_n)t}M_ne^{-(L_{s}+\frac{1}{n}J_n)t}\right) dt  \nonumber \\
 	& \leq &  -  \int_{0}^{\infty} \tr \left(M_ne^{-(L_{s}+L+\frac{2}{n}J_n)t} \right) dt  \nonumber \\
	&=& - \tr \left((L_{}+L_{s})^{\dag}\right ).
	\label{eq:1068}
	\end{eqnarray}
Putting all these pieces together, the $\mathcal H_2$-norm of the augmented network can be written as
	\begin{eqnarray}
	\| G^*\|^2_{\HH_2} &=& \tr \left (\begin{bmatrix} M_n & M_n \end{bmatrix} \begin{bmatrix} X_1 & X_2\\X_2 & X_3 \end{bmatrix}\begin{bmatrix} M_n \\M_n \end{bmatrix}\right )\nonumber \\
	&=&\tr \left (\begin{bmatrix} X_1 & X_2\\X_2 & X_3 \end{bmatrix}\begin{bmatrix} M_n & M_n\\M_n & M_n \end{bmatrix}\right ).
	\label{eq:1084}
	\end{eqnarray}
From \eqref{eq:1057}, \eqref{eq:1067} and \eqref{eq:1084}, it follows that
	\begin{eqnarray}
	\| G^*\|^2_{\HH_2} &=& \frac{1}{2}\tr (L_{}^{\dag}+L_{s}^{\dag})~-~\tr((X_2+X_2^{\text T})M_n)	\nonumber \\
	&\leq&\frac{1}{2}\tr (L_{}^{\dag}+L_{s}^{\dag})~-~2\tr \left ((L_{}+L_{s})^{\dag} \right), \nonumber 
	\end{eqnarray}
where \eqref{eq:1068} is used  in the last inequality. Finally, from this and \eqref{eq:924}, we conclude the desired result 
\[	\| G -G_s\|^2_{\HH_2} \, \leq \, \frac{1}{2}\tr \left (\hat{L}^{\dag}+L^{\dag}-4(L+\hat{L})^{\dag}\right ). \]
\end{proof}

The right-hand side of inequality \eqref{eq:908} is always non-negative, \ie  
\[ 0 ~\leq~ {\tr \left (\hat{L}^{\dag}+L^{\dag}-4(L+\hat{L})^{\dag}\right )}. \]
This is because of the fact that $\tr(L^{\dag})$ is convex on $\LL$ and the following inequality holds 
\[ \tr \left( (L +\hat L)^\dag \right) ~\leq~ \frac{1}{4} \tr \big(L^\dag \big)+ \frac{1}{4} \tr \big(\hat L ^\dag \big).\]

The inequality \eqref{proximity-cone} implies proximity of state matrices of the original and its abstraction on the cone of positive semidefinite matrices. In the following result, it is proven that the frequency specifications of two $\epsilon$-approximations are indeed very similar in $\mathcal H_2$ sense. 

\begin{figure}[t]
  \begin{center}  
  \psfrag{e}[t][t]{~~~{\color{blue}$\NN(L)$:}}   
  \psfrag{f}[t][t]{~~~{\color{blue}$\NN(L_s)$:}}   
  	 \psfrag{a}[t][t]{~~$\xi(t)$}      
	  \psfrag{b}[t][t]{~~~~~~~~$y(t)-y_s(t)$}   
	   \psfrag{c}[t][t]{$y(t)$}   
	    \psfrag{d}[t][t]{~~~$y_s(t)$}   
  \includegraphics[trim = 200 150 100 100, clip, width=0.5\textwidth]{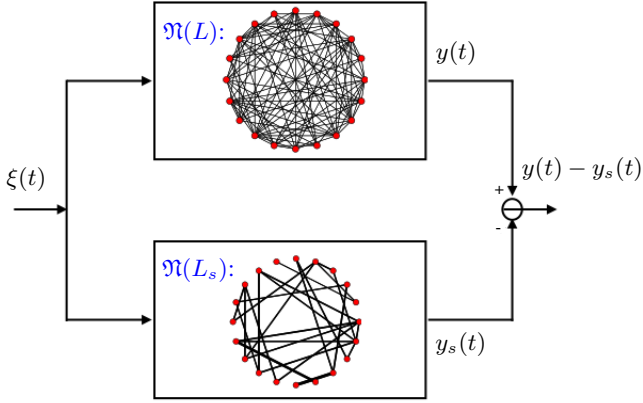}
   \end{center}
   \caption{\small{The block diagram of the augmented network \eqref{first-order-D}.}}
     \label{h_2_error_block}
\end{figure}

\begin{theorem}\label{prop-spars-norm}
If  $\NN(L_{s})$ is an $\epsilon$-approximation  of $\NN(L)$ for some $0 \leq \epsilon <1$, then
\begin{equation} \label{upper-bound}
\frac{\| G-G_s\|_{\mathcal H_2}}{\| G\|_{\mathcal H_2}} ~\leq~ {\sqrt{\frac{\epsilon(4-\epsilon)}{(1-\epsilon)(2+\epsilon)} }},
\end{equation}
where $G(s)$ and $G_s(s)$ represent the transfer matrices from input $\xi$ to output $y$ of $\NN(L)$ and $\NN(L_s)$, respectively. 
	\end{theorem}

\begin{proof}
According to the definition of $\epsilon$-approximation and the fact that $\tr(L^{\dag})$ is a homogeneous systemic performance measure of order $ -1$, we get
 \begin{equation}
\tr ( L_s^{\dag}) ~\leq~ \frac{1}{1-\epsilon} \tr (L^{\dag}).
\label{eq:757}
 \end{equation}
Moreover, it follows from Theorem \ref{th-abstract} that
 \begin{equation}
(2+\epsilon)^{-1} (L )^{\dag} ~ \leq ~ (L+ L_s)^{\dag}.
\label{eq:1113}
 \end{equation}
By taking trace from both sides of \eqref{eq:1113}, one obtains 
 \begin{equation}
 \tr \left( (L + L_s)^{\dag} \right) ~ \geq ~ \frac{1}{2+\epsilon}\tr \big(L^{\dag} \big).
 \label{eq:762}
 \end{equation}
By applying result of Lemma \ref{lemma-spars-norm}, we can conclude that 
\begin{eqnarray}
 \frac{\| G-{G_s}\|^2_{\mathcal H_2}}{\|G\|_{\mathcal H_2}^2}&\leq&\frac{\tr \left (\hat{L}^{\dag}+L^{\dag}-4(L+\hat{L})^{\dag}\right )}{\tr (L^\dag)}  \nonumber \\
 &=& \frac{\tr \left ({L_s}^{\dag}\right)+\tr \left (L^{\dag}\right)-4 \tr \left ( (L+{L_s})^{\dag} \right)}{\tr (L^\dag)} \nonumber \\
 & \leq & \frac{\epsilon(4-\epsilon)}{(1-\epsilon)(2+\epsilon)},
 \label{eq:769}
 \end{eqnarray}
where the last inequality is obtained after using  \eqref{eq:757} and \eqref{eq:762}. 
\end{proof}

  Fig. \ref{fig:approx} depicts the upper bound in inequality \eqref{upper-bound} for the relative $\mathcal H_2$-norm error of a linear consensus network and its $\epsilon$-approximation.


\begin{corollary}
Suppose that $y$ is the output of network $\NN(L)$ and $y_s$ is the output of its corresponding $\epsilon$-approximation $\NN(L_{s})$ for some $0 \leq \epsilon <1$. Then,  the steady-state expected value of their output error can be bounded by 
\begin{equation*}
 \lim_{t \rightarrow \infty}\E \Big\{\| y(t) - y_s(t)\|_2^2 \Big\} ~\leq~ \frac{\epsilon(4-\epsilon)}{2(1-\epsilon)(2+\epsilon)} \tr \big(L^{\dag} \big). 
 \end{equation*}
\end{corollary}
The proof of this corollary is based on the proof of Theorem  \ref{prop-spars-norm}.  

\begin{remark}
In \cite{Dhingra12}, the authors consider a similar problem of identifying a sparse representation of a given dense linear consensus network. Their proposed method has two steps. First, an optimal sparse network topology is obtained by adding the $\mathcal H_2$-norm error (c.f., Fig. \ref{h_2_error_block}) with another penalizing term that accounts for sparsity. Then, the optimal link weights are chosen over the identified topology. For both cases, approximation methods based on the Broyden-Fletcher-Goldfarb-Shanno method are employed in \cite{Dhingra12}. This method chooses a descent direction based on an approximation of the Hessian matrix. Therefore each update requires $\mathcal O(n^6)$ operations\cite{Dhingra12}. In comparison with Algorithm 1 in Section \ref{spars}, the proposed method in \cite{Dhingra12} is computationally expensive. Furthermore, it does not provide any guaranteed performance certificates. 
\end{remark}

\begin{figure}[t]
	\centering
	 \psfrag{x}[t][t]{ \footnotesize{$\epsilon$}}        
       \psfrag{y}[b][b]{\footnotesize{${\sqrt{{\epsilon(4-\epsilon)}/{(1-\epsilon)(2+\epsilon)} }}$}} 
	\includegraphics[trim = 5 5 5 5, clip,width=.48 \textwidth]{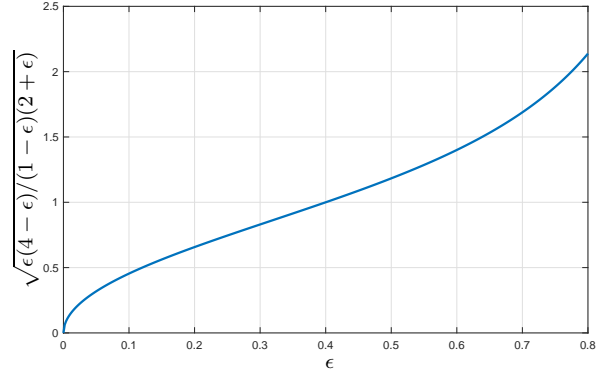} 
	\caption{\small This plot presents the upper bound given by Theorem \ref{prop-spars-norm} on the $\mathcal H_2$-norm error of a consensus network and its $\epsilon$-approximation network.}
	\label{fig:approx}
\end{figure}

\section{Localized Network Abstraction} \label{sec:partial}

Our methodology can be extended further to explore several interesting network design problems, such as partial or localized abstraction of a given large-scale consensus network. In this section, we only look at one of such design problems. Let us consider a slightly modified version of \eqref{first-order-G} by involving a predesigned  state feedback controller 	\begin{eqnarray}
		\NN(L_0+L_1):\begin{cases}
			\dot x(t)\,=\, -L_{0}x(t)\,+\,u(t)\,+\,\xi(t)\\
			u(t)\,=\,-L_{1}x(t)\\
			y(t)\,=M_n x(t)
		\end{cases}
	\label{modified-consensus-algorithm}
	\end{eqnarray}
with initial condition $x(0) =x_0$, where $L_0$ is the Laplacian matrix of the open-loop network and the Laplacian matrix $L_1$ is the predesigned state feedback gain matrix. Let us represent the corresponding coupling graph to $L_1$ by $\G_1$. If $L_1$ is obtained via traditional optimal control methods without incorporating sparsity measures, then one should expect to get a dense  interconnection topology for $\G_1$; we refer to  \cite{motee-sun} for discussions on spatially decaying structure of optimal controllers. Therefore, our design objective is to compute a localized abstraction for the closed-loop network $\NN(L_0+L_1)$ that only sparsifies $\G_1$. Let us represent such an abstraction  by $\hat{\G}_1$ with Laplacian $\hat{L}_1$. 

\begin{theorem}\label{partial}
Suppose that a linear consensus network with structure (\ref{modified-consensus-algorithm}), a homogeneous systemic performance measure $\rhoo:\LL \rightarrow \R_+$ of order $-\alpha$, and a design parameter $d > 2$ are given. For $\epsilon= \frac{\sqrt{8d}}{d+2}$, there exists a subgraph abstraction $\hat{\G}_1=(\V,\hat{\EE},\hat{w})$ of $\G_1=(\V,\EE,w)$ with at most $dn/2$ links that satisfies $\hat{\EE} \subset \EE$ and  
		\begin{equation}
			\left| \frac{\Pi_{\rho}(L_{0}+{L}_1)-\Pi_{\rho}(L_{0}+\hat L_{1})}{\Pi_{\rho}(L_{0}+\hat L_{1})}	\right|\leq \epsilon. 
		\end{equation}
Furthermore, it follows that
		\[ w_{\text{total}}(\hat{L}_1) ~\leq ~(1+\epsilon) \hspace{0.05cm} w_{\text{total}}(L_1) \]
{ in which $w_{\text{total}}(L) = \frac{1}{2} \tr(L)$.}		
	\end{theorem}
\begin{proof}
According to \cite[Th. 1.1]{Spielman14}, coupling graph $\G_1=(\V, \EE,w)$ has a weighted subgraph $\mathcal H=(\V, \hat \EE, \hat w)$ with $|\hat \EE|=\lceil {d(n-1)/2} \rceil$ that satisfies
\begin{equation}
L_1 \, \preceq \, L_{\mathcal H} \, \preceq \, \frac{d+2+\sqrt{8d}}{d+2-\sqrt{8d}} ~L_1
\label{eq:671}
\end{equation}
where $L_{\mathcal H}$ is the Laplacian matrix of graph $\mathcal H$. We define $\hat {\mathcal G}_1=(\V, \hat \EE, \hat w)$ based on the following  Laplacian matrix  
\begin{equation}
\hat L_1 ~:=~  \frac{(1-\sqrt{d/2})^2}{1+d/2} ~L_{\mathcal H}.
\label{eq:676}
\end{equation}
From \eqref{eq:671} and \eqref{eq:676}, it follows that 
\begin{equation}
 \left(1-\frac{\sqrt{2d}}{d/2 +1}\right) L_1 \, \preceq \, \hat L_{1}  \, \preceq \,  \left(1+\frac{\sqrt{2d}}{d/2 +1}\right)L_1.
 \label{eq:681}
\end{equation}
Moreover, we know that
\begin{equation}
 \left(1-\frac{\sqrt{2d}}{d/2 +1}\right) L_0 \, \preceq \,  L_{0}  \, \preceq \,  \left(1+\frac{\sqrt{2d}}{d/2 +1}\right) L_0.
 \label{eq:686}
\end{equation}
From \eqref{eq:681} and \eqref{eq:686}, we have
\begin{eqnarray}
\left(1-\frac{\sqrt{8d}}{d +2}\right) (L_0 + L_1) & \preceq &  (L_{0} + \hat L_1) \nonumber \\
 &\preceq & \left(1+\frac{\sqrt{8d}}{d +2}\right) (L_0+L_1).
 \label{eq:691}
\end{eqnarray}
Using \eqref{eq:691} and  the result of Theorem \ref{th-abstract}, it yields that $\NN(L_0+\hat L_{1})$ is a $(\frac{\sqrt{8d}}{d+2}, d)$-abstraction of $\NN(L_0+L_1)$. 
\end{proof}
	
\vspace{0.1cm}
This result is particularly useful is abstraction of large-scale consensus networks where the control objective is to abstract only desired parts of a network without drastically affecting the global performance.

\section{Parallel Network Abstraction}
\label{sec:parallel}

Building upon the results of the previous section, we introduce a distributed and parallel implementation of our proposed localized algorithm. The main advantage of parallel abstraction is that several localized abstraction problems can be solved simultaneously by dividing a large-scale network abstraction problem  into smaller in size localized problems.

\vspace{.0cm}
\begin{definition}
A base subgraph $\mathcal G_0$ of a network is a subgraph which is formed by those feedback links that will stay unchanged throughout the abstraction process. \end{definition}
\vspace{.0cm}

Let us denote Laplacian matrix of a base subgraph $\mathcal G_0$  by $L_0$.  For a given natural number $p$, the original coupling graph can be partitioned as the union of a base subgraph and $p$ other subgraphs $\G_1, \ldots, \G_p$, \ie 
\begin{equation}
L~=~ L_0 + \sum_{i=1}^p L_i,
\label{eq:852}
\end{equation}
where $L_i$ is Laplacian matrix of subgraph $\G_i$. Without loss of generality, one may assume that the node set of all subgraphs is $\V$, where  $\V$ is the node set of the original graph. This assumption implies that the corresponding Laplacian matrices are compatible $n$-by-$n$ matrices. The network setup for parallel implementation of abstraction takes the following form 	
\begin{eqnarray}
\begin{cases}
			\dot x(t)\,=\, -L_{0}x(t)\,+\,u(t)\,+\,\xi(t)\\
			\displaystyle u(t)\,=\,-\sum_{i=1}^p L_i \hspace{0.05cm} x(t)\\
			y(t)\,=M_n x(t)
		\end{cases}
	\label{parallel-algorithm}
	\end{eqnarray}
with initial condition $x(0) =x_0$.

\vspace{.0cm}
\begin{assumption}
The corresponding  subgraphs to $L_i$ for $i=1,\ldots,p$ are link-disjoint and dense.
\end{assumption}

\vspace{.0cm}
\begin{theorem}
	\label{parallel}
Suppose that a linear consensus network with structure (\ref{parallel-algorithm}), a homogeneous systemic performance measure $\rhoo:\LL \rightarrow \R_+$ of order $-\alpha$, and a number $d > 2$ are given. If the coupling graph of the network can be decomposed as \eqref{eq:852}, then for $\epsilon= \frac{\sqrt{8d}}{d+2}$  there exists a set of subgraph sparsifier $\{\hat{\G}_i \}_{i=1}^p$ for $\{\G_i \}_{i=1}^p$ where each sparsifier subgraph has at most $dn/2$ links and the global performance index satisfies 
		\begin{equation}
			\left| \frac{\Pi_{\rho}(L_{0}+\sum_{i=1}^p{L}_i)-\Pi_{\rho}(L_{0}+\sum_{i=1}^p\hat L_{i})}{\Pi_{\rho}(L_{0}+\sum_{i=1}^p\hat L_{i})}	\right|~\leq~ \epsilon. 
		\end{equation}
Furthermore, it follows that
		\[ w_{\text{total}}(\hat{L}_i) ~\leq ~(1+\epsilon) \hspace{0.05cm} w_{\text{total}}(L_i), \]
for $i=1,2,\cdots,p$.		
	\end{theorem}
\begin{proof}
For every $i=1, \cdots, p$, coupling subgraph $\G_i=(\V, \EE_i,w_i)$ has a weighted subgraph $\mathcal H_i=(\V, \hat \EE_i, \hat w_i)$ with $|\hat \EE_i|=\lceil {d(n-1)}/2 \rceil$ whose Laplacian satisfies
\begin{equation}
L_i ~\preceq~ L_{\mathcal H_i} ~\preceq~ \frac{d+1+2\sqrt{d}}{d+1-2\sqrt{d}} L_i
\label{eq:671b}
\end{equation}
where $L_{\mathcal H_i}$ is the Laplacian matrix of graph $\mathcal H_i$ \cite[Th. 1.1]{Spielman14}. We define every $\hat {\mathcal G}_i=(\V, \hat \EE_i, \hat w_i)$ according to the following  Laplacian matrix
\begin{equation} 
\hat L_i ~:=~   \frac{(1-\sqrt{d/2})^2}{1+d/2}  L_{\mathcal H_i}.
\label{eq:676b}
\end{equation}
By putting \eqref{eq:671b} and \eqref{eq:676b} together, it follows that 
\begin{equation}
 \left(1-\frac{\sqrt{8d}}{d +2}\right)  L_i ~\preceq~ \hat L_{i} ~\preceq~   \left(1+\frac{\sqrt{8d}}{d +2}\right)  L_i.
 \label{eq:681b}
\end{equation}
Moreover, we know that
\begin{equation}
 \left(1-\frac{\sqrt{8d}}{d +2}\right)  L_0 ~\preceq~  L_{0} ~\preceq~   \left(1+\frac{\sqrt{8d}}{d +2}\right)  L_0.
 \label{eq:686b}
\end{equation}
From \eqref{eq:681b} and \eqref{eq:686b}, we get
\begin{eqnarray}
&& \hspace{-1.4cm} \left(1-\frac{\sqrt{2d}}{d/2 +1}\right)  (L_0 + \sum_{i=1}^p L_i) \, \preceq L_{0} + \sum_{i=1}^p \hat L_i  \nonumber \\
&&~~~~~~~~~~~~~~~~~~~~ \preceq \,  \left(1+\frac{\sqrt{2d}}{d/2 +1}\right)  (L_0+\sum_{i=1}^p L_i).
 \label{eq:691b}
\end{eqnarray}
From \eqref{eq:691b} and  the result of Theorem \ref{th-abstract}, one can conclude that  $\NN(L_0+\sum_{i=1}^p \hat L_{i})$ is a $(\frac{\sqrt{8d}}{d+2}, d)$-abstraction of $\NN(L_0+\sum_{i=1}^pL_i)$. 
\end{proof}

\vspace{0.1cm}

Algorithm 1 can be implemented on several parallel localized processing units to abstract every $\G_i$ for $i=1,\cdots,p$. This parallelization scheme cuts the time complexity of solving an abstraction problem down to  ~$\tilde{\mathcal O} (m_{\max})$, where $m_{\max}$ is the number of links of the densest graph among $\G_1, \ldots, \G_p$.

\section{Shortcomings of $\ell_0$/$\ell_1$-Regularized Sparsification Methods}\label{ell0-ell1}
In order to put our proposed methodology into perspective, we discuss some of the shortcomings of the $\ell_0$/$\ell_1$-regularization based sparsification methods. The common approach is to formulate an optimal control problem that is augmented by a penalty term to promote sparsity. The resulting optimal control problem can be usually cast as a bilinear matrix optimization problem and convexified using alternating methods \cite{farlinjovTAC14sync, arstoo-cdc2016}. While $\ell_0$/$\ell_1$-regularization based methods generally do not scale with network size and suffer from high time complexities, we would like to accentuate a more important issue by means of an example that shows $\ell_0$/$\ell_1$-regularized methods do not always return sparse solutions. Let us consider a linear consensus network that is governed by \eqref{first-order-G} whose feedback structure is represented by a complete graph with identical link weights $w_0$, Laplacian matrix $L_0$, link set $\mathcal{E}_0$, and incidence matrix $E_0$. The control objective is to eliminate or reweigh some links of the network while minimizing $\mathcal{H}_2$-norm of the network from $\xi$ to $y$. This problem can be formulated as the following $\ell_0$-problem   
\begin{eqnarray}
& & \hspace{-5.1cm}\underset{w(e)}{\textrm{Minimize}}  ~~\|G_F\|^2_{\mathcal{H}_2}~+~{\frac{\gamma}{2}} \hspace{0.05cm} \|A_F\|_{\ell_0} \label{counter-1}\\
& & \hspace{-5.1cm} \mbox{subject to:} \nonumber\\
& & \hspace{-5.1cm} \displaystyle F~=~\sum_{e \in \mathcal{E}_0}w(e) b_e  b_e^{\rm T} \label{counter-2}\\
& & \hspace{-5.1cm} w(e) \leq w_0 \label{counter-2.5} \\
& & \hspace{-5.1cm} L_0-F \in \LL \label{counter-3}
\end{eqnarray}
in which $G_F(s)$ is the transfer matrix of network $\NN(L_0-F)$ from $\xi$ to $y$, $A_F$ is the adjacency matrix of  $\NN(L_0-F)$,  and $b_e$ is the column of incidence matrix $E_0$ that corresponds to link $e$. Boundedness of the cost function as well as constraint \eqref{counter-3} ensures connectivity of the coupling graph and  positive semi-definiteness of the Laplacian matrix of  $\NN(L_0-F)$. In the penalty term, $\gamma >0$ is a design parameter and $\ell_0$ sparsity measure of matrix $A_F$ is equal to the total number of nonzero elements of $A_F$, which is  defined by \eqref{sparsity-measure-1}. For an undirected graph, the value of this sparsity measure is equal to twice the number of the edges. The value of cost function \eqref{counter-1} is greater than or equal to $\gamma (n-1)$. One can get arbitrarily close to this lower bound through the following steps. Let us select a spanning tree $\mathcal{T}_0$ from link set $\mathcal{E}_0$ and set elements of $F$ in a way that $L_0-F$ becomes the Laplacian matrix of $\mathcal{T}_0$ with link weights $w_0 - w(e)$. It is known that quantity $\|G_F\|_{\mathcal{H}_2}$ is bounded if and only if the corresponding coupling graph to $L_0-F$ has at least $n-1$ links \cite{Siami14arxiv}, which in our case $\mathcal{T}_0$ does have $n-1$ links. Therefore,  the value of $\|G_F\|_{\mathcal{H}_2}$ can be made arbitrarily close to $0$ by letting $w(e)$ get closer to  $-\infty$. 

After relaxing $\ell_0$-measure by $\ell_1$-norm and representing the $\mathcal{H}_2$-norm in terms of Laplacian spectrum, the objective function \eqref{counter-1} can be rewritten as 
\begin{equation}
\frac{1}{2} \sum_{i=2}^n \lambda_i(L_0-F)^{-1}~+~\frac{\gamma}{2} \hspace{0.05cm} \|A_F\|_{\ell_1} \label{cost-1}
\end{equation}
where $\|A_F\|_{\ell_1}$ is equal to the sum of the absolute values of all elements of $A_F$ and 
\[ \|A_F\|_{\ell_1} = 2 \sum_{e \in \mathcal{E}_0}(w_0-w(e)) = \tr(L_0-F). \]
The cost function \eqref{cost-1} can be simplified further to obtain
\begin{equation}
\frac{1}{2} \sum_{i=2}^n \lambda_i(L_0-F)^{-1}~+~\frac{\gamma}{2} \hspace{0.05cm} \sum_{i=2}^n \lambda_i(L_0-F). \label{cost-2}
\end{equation}

Thus, the $\ell_1$-problem is the minimization problem with cost function \eqref{cost-2} and constraints \eqref{counter-2}-\eqref{counter-3}. This problem is convex and has a unique solution. This follows from the following facts that $\mathcal{H}_2$-norm is a systemic measure and convex, trace is a linear operator and convex, and sum of these two convex functions results in a convex cost function  \eqref{cost-2}. The inequality \eqref{counter-2.5} also represents a convex set.  Moreover, positive semi-definiteness of $L_0-F$ and connectivity of the resulting network $\NN(L_0-F)$ are both convex constraints; therefore, constraint \eqref{counter-3} is convex. By applying the arithmetic and geometric means inequality, one can show that the spectral function \eqref{cost-2} is lower bounded by constant $(n-1) \sqrt{2 \gamma}$. By respecting all the constraints and for all design parameters $\gamma > \frac{1}{8 n^2 w_0^2}$, this lower bound can be achieved by a complete graph with identical link weights {$w(e)^*=\frac{1}{n \sqrt{2 \gamma}}-w_0$}.  An interesting scenario happens when $\gamma=2$ and $w_0 > \frac{1}{4}n^{-1}$ where the lower bounds on the least achievable cost values for both $\ell_0$-problem and its relaxed $\ell_1$-problem coincide and become $2 (n-1)$. In this case,  the optimal networks from solving  the $\ell_0$-problem and $\ell_1$-problem have very different sparsity patterns: the $\ell_0$-problem solution is a network with a spanning tree topology and arbitrarily large link weights, while $\ell_1$-problem identifies a network whose coupling graph is complete with identical link weights. This shows that $\ell_1$-relaxation of a $\ell_0$-regularized optimal control does not always provide sparse solutions.

\section{ Illustrative Examples}
\label{example}

\begin{figure}[t]
\centering
\begin{tabular}{c}
\includegraphics[trim = 45 65 35 65, clip,width=.43 \textwidth]{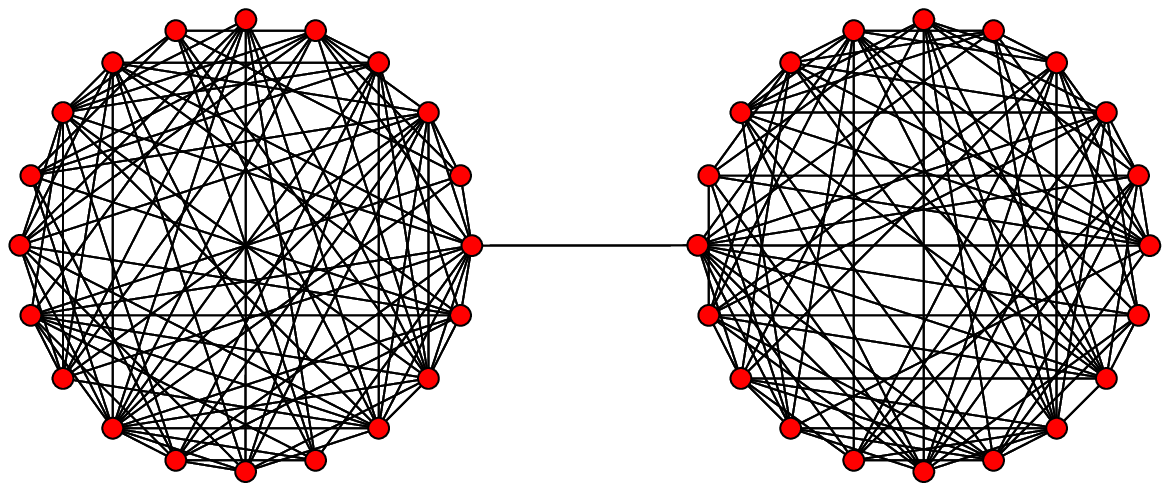}\\(a)\\
 \includegraphics[trim = 45 65 35 65, clip,width=.43 \textwidth]{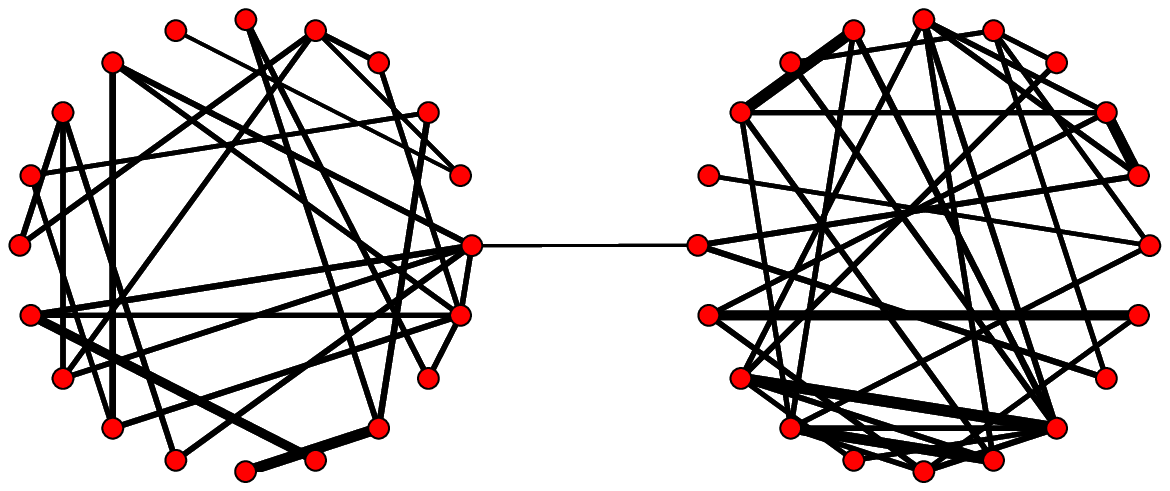} \\(b)\\
\end{tabular}
\caption{\small (a) An unweighted coupling graph with $40$ nodes, $201$ links, $\|G\|_{\mathcal H_2}= 2.7837$, and $w_{\text{total}}(L)=201$.  (b) An abstraction of Fig. \ref{40_node}(a) with $40$ nodes, $61$ links, $\|G_s\|_{\mathcal H_2}=3.0805$, and $w_{\text{total}}(L_s)=199.88$.}
\label{40_node}
\end{figure}

In this section, we present several numerical examples to illustrate our theoretical findings. 

\begin{figure}[t]
\begin{center}
\psfrag{y}[c][c]{\small Probability}
\psfrag{x}[c][c]{\small Label of links}
 \includegraphics[trim = 7 7 7 7, clip,width=.52 \textwidth]{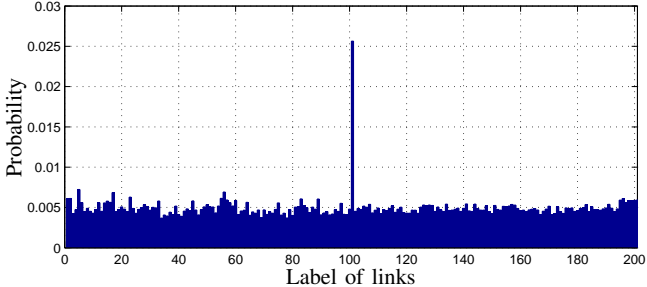} 
\caption{\small The probability of selecting a link  of the coupling graph given by Fig. \ref{40_node}(a) as an important link.  A cut edge is the most important link with the highest probability: if we throw out the cut edge, the coupling graph of the resulting abstraction will be disconnected for sure.}
\label{probability_A}
\end{center} 
\end{figure}
%
	

\begin{example}
\label{40node:example}
We first consider a consensus network with $40$ agents defined over an unweighted coupling graph with two dense components which are connected by a single link, i.e., a cut edge. Each of the components is obtained by adding $100$ uniformly and randomly selected  links to an empty graph with $20$ nodes; see Fig. \ref{40_node}(a). 
Based on Algorithm 1, we sample low-connectivity coupling links (\ie~ feedback gains) with high probability and high-connectivity coupling links with low probability. The probability of selecting a link from the coupling graph is depicted in Fig. \ref{probability_A}. One observes that the probability of selecting the cut edge as an important link is much higher than the probability of choosing other links. Fig. \ref{40_node}(b) shows a (0.5, 3.05)-abstraction of the network after applying Algorithm 1 that has $61$ links and meets all requirements of Definition \ref{abstract-def}. The coupling graph of the abstraction is weighted and has about 70\% fewer links than the original network.  Although, we set $\epsilon=0.5$  when running Algorithm 1, the performance loss of the resulting abstraction is less than $24$\% according to Table \ref{table_2}.

\begin{table}[t]
\centering
{\small
\begin{tabular}{ |l| c|  }  
\hline
Systemic Performance Measure & $\frac{ |{\Pi_\rhoo(L_s)}-{\Pi_\rhoo(L)}|}{\Pi_\rhoo(L_s)} \times 100$    \\
\hline
\hline
System Hankel Norm: $\frac{1}{2}{\lambda_2}^{-1}$& 19.65 \%  \\ 
\hline
 Squared $\mathcal H_2$ Norm: $\frac{1}{2}\sum_{i=2}^n{\lambda_i}^{-1}$&  18.34 \%   \\ 
\hline
Zeta Spectral Norm: $\Big (\sum_{i=2}^n{\lambda_i}^{-2}\Big)^{\frac{1}{2}}$& 15.26 \% \\ 
\hline
Local Deviation Error: $\frac{1}{2}\sum_{i \in \mathcal V} {d_i}^{-1}$& 23.16  \% \\ 
\hline
\end{tabular}}
\caption{\small{Relative performance loss percentage of network Fig. \ref{40_node}(a) with respect to its abstraction Fig. \ref{40_node}(b).}}
\label{table_2}
\end{table}

\end{example}

       \begin{figure}[t]
	\centering
	\begin{tabular}{c c}
	\includegraphics[trim = 45 20 35 10, clip,width=.225 \textwidth]{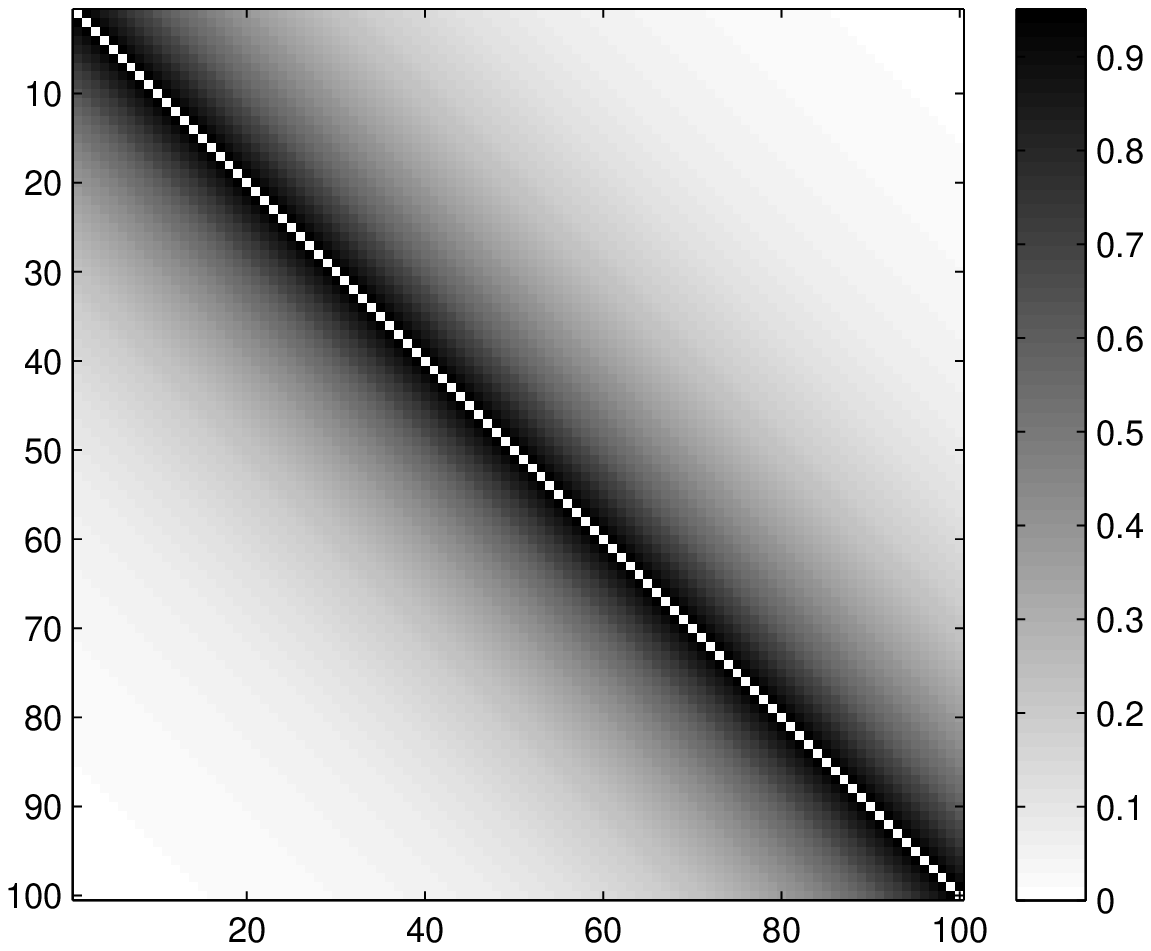} &
	\includegraphics[trim = 45 20 35 10, clip,width=.225 \textwidth]{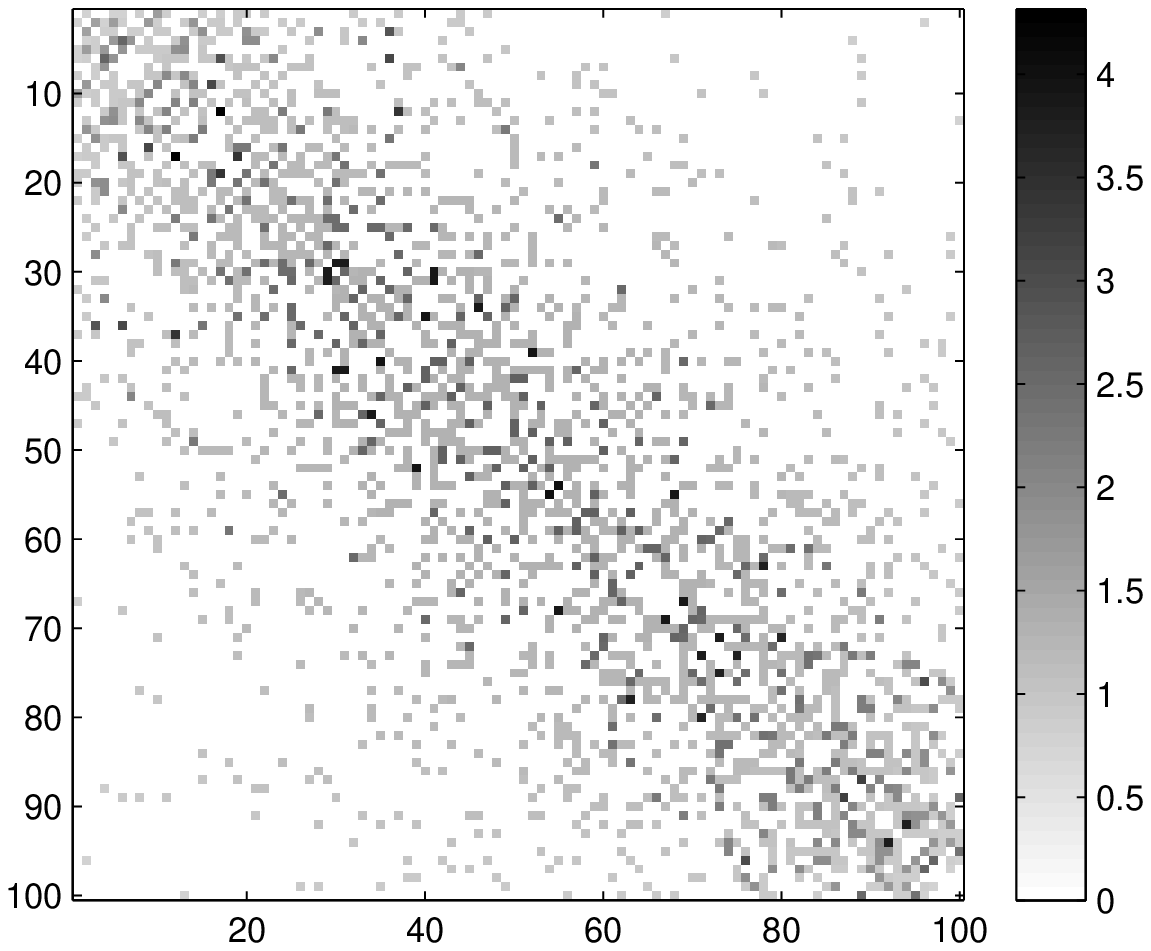} \\ (a) & (b)
	\end{tabular}
	\caption{\small (a) This plot demonstrates the sparsity pattern of adjacency matrix of network in Example \ref{1:example} with coupling  parameters $c=1$ and $\gamma=0.05$. This network has $100$ agents and $4,950$ links and the color intensity at each pixel of this plot shows magnitude of the corresponding element in the adjacency matrix. (b) This plot depicts the sparsity pattern of adjacency matrix of an abstraction that has $1114$ links. The relative $\mathcal H_2$ error of these two networks is $\|G-G_s\|_{\mathcal H_2}/\|G\|_{\mathcal H_2}=0.18$ and the ratio of their total weights is $w_{\text{total} }(L_s)/w_{\text{total} }(L) =1.0028$.} \label{decay} 
\end{figure}
\begin{table}[t]
        \centering
        {\small
        \begin{tabular}{ |l| c|  }  
        \hline
        Systemic Performance Measure~&~ $\frac{ |{\Pi_\rhoo(L_s)}-{\Pi_\rhoo(L)}|}{\Pi_\rhoo(L_s)} \times 100$ ~\\ 
        \hline
        \hline
        System Hankel Norm: $\frac{1}{2}\lambda_2^{-1}$& 10.72 \% \\  
        \hline
        Squared $\mathcal H_2$ Norm: $\frac{1}{2}\sum_{i=2}^n{\lambda_i}^{-1}$&  6.44 \% \\ 
        \hline
        Zeta Spectral Norm: $\Big (\sum_{i=2}^n{\lambda_i^{-2}}\Big)^{\frac{1}{2}}$& 9.69 \% \\ 
        \hline
        Local Deviation Error: $\frac{1}{2}\sum_{i \in \mathcal V} {d_i^{-1}}$& 3.07 \% \\ 
        \hline
        \end{tabular}}
        \caption{\small{Relative performance loss percentage of network Fig. \ref{decay}(a) with respect to its abstraction Fig. \ref{decay}(b) that has $77.49$ \% fewer links.}}
        \label{table_3}
\end{table}

\begin{example}
\label{1:example}
Let us consider a consensus network with 100 agents and exponentially decaying couplings that are defined by
\begin{figure}[t]
	\centering
	\includegraphics[trim = 5 5 5 5, clip,width=.52 \textwidth]{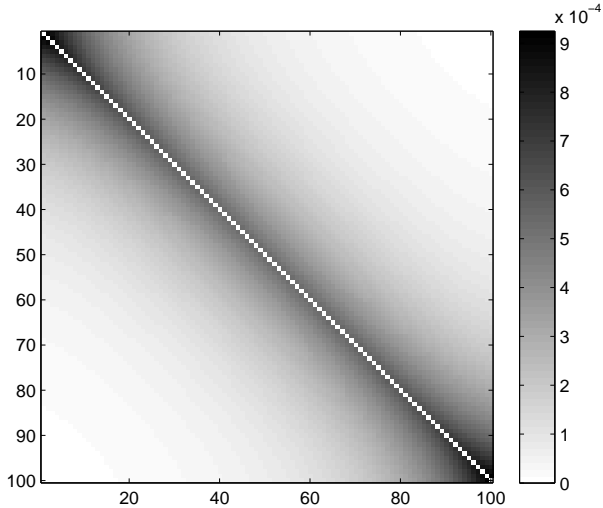} 
	\caption{\small This plot presents the probability distribution of the sampling process to choose important links in Example \ref{1:example}. The color intensity of each pixel shows importance of that link for sampling.}
	\label{prob}
\end{figure}

	\begin{eqnarray}
		w(\{i,j\}) ~=~ \left\{ 
		\begin{array}{l l}
		c \hspace{0.05cm}e^{-\gamma |i-j|} & \quad \text{if $i \neq j$}\\
		0 & \quad \text{if $i=j$}
		\end{array} \right.
		\label{spa-dec}
	\end{eqnarray}  
where $c$ and  $\gamma$ are positive numbers and $i, j \in \mathcal V$. This class of networks arises in various applications where there is a notion of spatial distance between the subsystems; we refer to  \cite{motee-sun} for more details.  
Fig. \ref{decay}(a) shows the adjacency matrix of the coupling graph of this network.
According to Theorem \ref{th-main}, this network has a $(0.5, 27.85)$-abstraction.   Fig. \ref{decay}(b) illustrates a $(0.5, 22.28)$-abstraction of the network after applying Algorithm 1, where the design parameter $\epsilon$ is set to $0.5$. The original network has $4,950$ links, while its abstraction has $1,114$ coupling links and meets all requirements of Definition \ref{abstract-def}. As a result, abstraction achieves $77.49$ \% sparsification. Although, in this example, we allow $50$\% performance loss, numerous simulation examples assert that the resulting bounds for performance loss can be comparably smaller.  As it is seen from Table \ref{table_3}, the relative performance loss percentage is less than $11$\%. Fig. \ref{prob} depicts the probability distribution of sampling process to select important links in Algorithm 1, where  the color intensity of pixels show how important that link is. According to Algorithm 1, low-connectivity coupling links are sampled with higher probability than  high-connectivity coupling links.

\begin{figure}[t]
                \centering
                \begin{tabular}{c }
                \includegraphics[trim = 45 20 40 0, clip,width=.43 \textwidth]{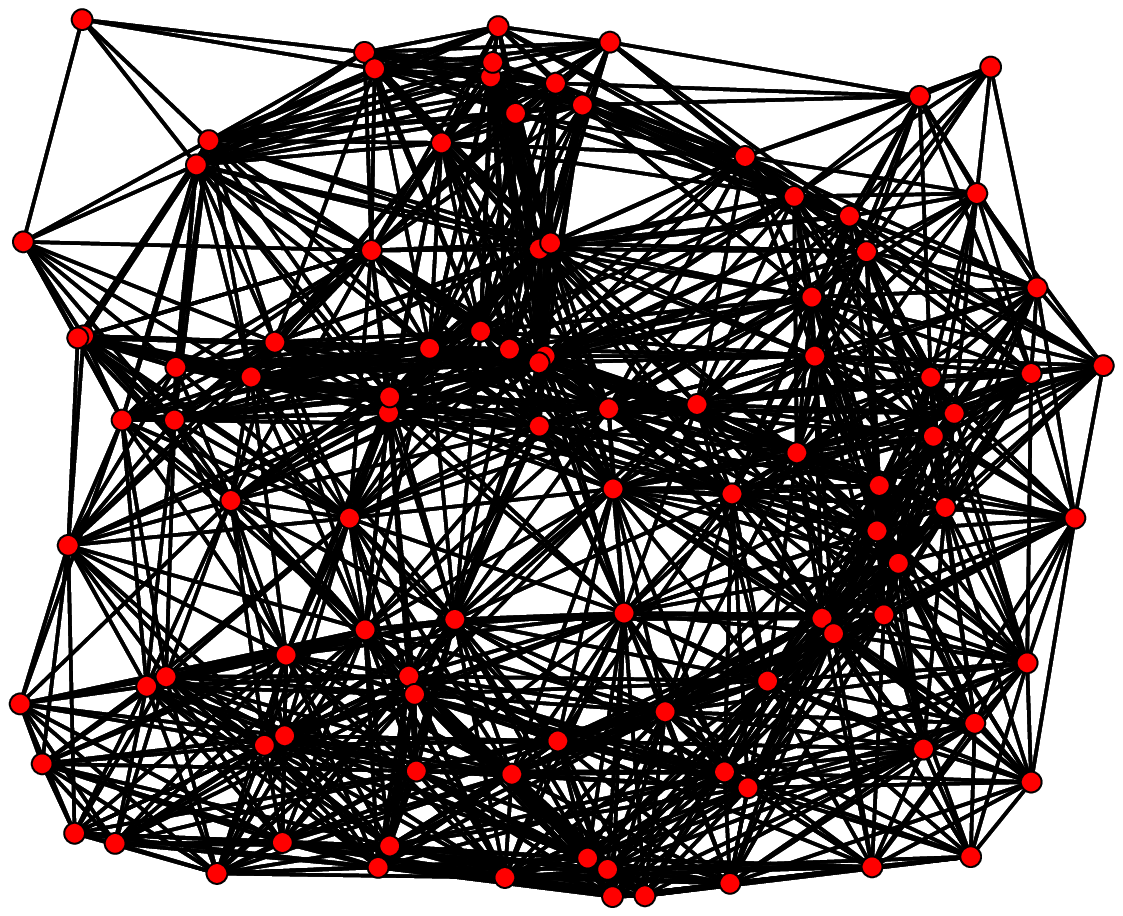} \\
                (a)\\
                \includegraphics[trim = 45 20 40 0, clip,width=.43\textwidth]{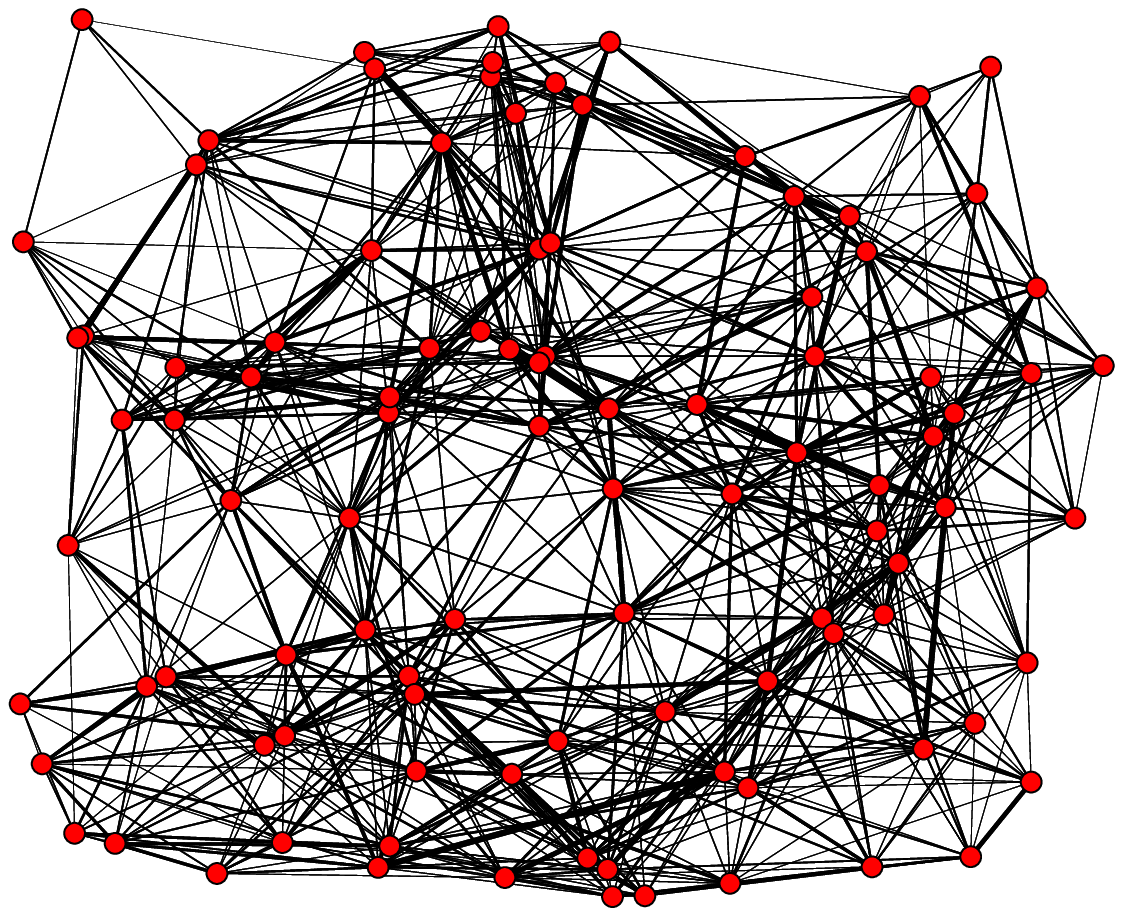} \\
                (b)
                \end{tabular}
                \caption{\small {(a)} An unweighted coupling (proximity) graph of a consensus network  with $100$ agents is presented. Every agent is connected to all of its spatial neighbors within a closed ball of radius $r=10$. This graph has $1,291$ links and  $w_{\text{total}}(L)=1,291$. (b) This graph shows a $(0.5,16.62)$-abstraction of the network. Our abstraction algorithm results in a network with a weighted coupling graph that has $831$ links and $w_{\text{total}}(L_s)=1293.4$. The relative $\mathcal H_2$ error of these two networks is $\|G-G_s\|_{\mathcal H_2}/\|G\|_{\mathcal H_2}=0.17$.}
         \label{ex_3}
	\end{figure}

\end{example}

\begin{example}
\label{2:example}
	
Let us consider a dynamical network consists of $100$ agents which are randomly distributed in a $30 \times 30$ square-shape area in space and  are coupled over a proximity graph. Every agent is connected to all of its spatial neighbors within a
closed ball of radius $r=10$. Fig. \ref{ex_3}(a) shows the resulting coupling graph of this dynamical network that has $100$ nodes and $1,291$ links and Fig.~\ref{ex_3}(b) depicts an$(0.5,16.62)$-abstraction of this network, which is obtained using Algorithm 1 with $\epsilon = 0.5$. 
The number of coupling links in this abstraction is $831$, which is $35.63\%$ sparsification, and meets all requirements of Definition \ref{abstract-def}. Table \ref{table_4} summarizes the percentage of the relative performance loss with respect to some systemic performance measures.

\begin{table}[h]
	\centering
	{\small
	\begin{tabular}{ |l| c|   }  
	\hline
	Systemic Performance Measure~&~$\frac{ |{\Pi_\rhoo(L_s)}-{\Pi_\rhoo(L)}|}{\Pi_\rhoo(L_s)} \times 100~$  \\  
	\hline
	\hline
	Hankel Norm: $\frac{1}{2}{\lambda_2}^{-1}$& 19.65 \%  \\  
      \hline
      Squared $\mathcal H_2$-Norm: $\frac{1}{2}\sum_{i=2}^n{\lambda_i}^{-1}$&   18.34 \% \\ 
      \hline
      Zeta Spectral Norm: $\Big (\sum_{i=2}^n{\lambda_i^{-2}}\Big)^{\frac{1}{2}}$& 15.26 \%\\ 
      \hline
      Local Deviation Error: $\frac{1}{2}\sum_{i \in \mathcal V} {d_i^{-1}}$& 23.16 \%\\ 
      \hline
      \end{tabular}}
      \caption{\small{Relative performance loss percentage of network Fig. \ref{ex_3}(a) with respect to its abstraction Fig. \ref{ex_3}(b) that has  $35.63\%$ fewer links.}}
      \label{table_4}
\end{table}

\end{example}

\begin{example}
\label{form:example}
Let us consider a simple model for the formation control of a group of autonomous vehicles, which is given by 
\begin{eqnarray}
\left[ \begin{array}{ccc}
\dot x(t) \\
\dot v(t) \end{array} \right] & =& \left[ \begin{array}{ccc}
0& I  \\
		-L & -\beta L  \end{array} \right] \left[ \begin{array}{ccc}
x(t) \\
		v(t) \end{array} \right]+\left[ \begin{array}{ccc}
0\\
		I \end{array} \right] \xi(t) \label{formation-1}
		\\
 y(t) &=& M_n v(t) \label{formation-2}		
\end{eqnarray}
where $\beta > 0$ is a design parameter. Each vehicle has a position and a velocity variable. The state variable of the entire network is denoted by $[\begin{array}{cc}x(t) & v(t)\end{array}]^{\text T}$ and is measured relative to a pre-specified desired trajectory $x_d(t)$ and velocity $v_d(t)$. Without loss of generality, we may assume that the position and velocity of each vehicle are scalar variables. The reason is that one can decouple  higher $D$-dimensional models into $D$ decoupled  \eqref{formation-1}-\eqref{formation-2} models. The overall objective is for the network to reach a desired formation pattern, where each autonomous vehicle travels at the constant desired velocity $v_d$ while preserving a pre-specified distance between itself and each of its neighbors. In this model, the state feedback controller uses both position and velocity measurements and $L$ is, in fact, the corresponding feedback gain, which represents the coupling topology in the controller array, and constant $\beta$ is a design parameter  \cite{Bamieh11, Bamieh12}.
We consider the steady-state variance of the performance output of this network  as the performance measure. This quantity is indeed equivalent to the square of the $\mathcal H_2$-norm of the system from the exogenous disturbance input  to the performance output \cite{Bamieh12, Young10, Jadbabaie13, Bamieh11, bamgay13acc,LovisariGarinZampieriResistance}. The squared $\mathcal H_2$-norm of \eqref{formation-1}-\eqref{formation-2} can be characterized in terms of Laplacian eigenvalues of the coupling graph as follows
	\begin{equation}
	\Theta_2^2(L):=\lim_{t \rightarrow \infty} \mathbb E \left[ y^{\text T}(t)y(t) \right] \, = \, \frac{1}{2\beta}\sum_{i=2}^n \lambda_i^{-2} \, = \, \frac{1}{2\beta} \, \zeta_2^2(L).\label{perf-meas}
	\end{equation}
We refer to \cite{Siami13cdc} for more details. This quantity is a homogeneous systemic performance measure; therefore, we can apply our abstraction algorithm.  Suppose that the coupling graph of network \eqref{formation-1}-\eqref{formation-2} is given by Fig. \ref{ex_3}(a). Then, as we mentioned in Example \ref{2:example}, Fig.~\ref{ex_3}(b) illustrates one example of $(0.5,16.62)$-abstraction of this network that is obtained using Algorithm 1. As shown in Example \ref{2:example}, this network has fewer coupling links ($831$ links) compare to the original network. Table \ref{table_5} presents the percentage of the relative performance loss, where both systemic performance measures are homogeneous of order  $-2$. This example shows that our proposed abstraction algorithm can be successfully applied to second-order linear consensus networks as well.  

\begin{table}[h]
	\centering
	{\small
	\begin{tabular}{ |c| c|   }  
	\hline
	Systemic Performance Measure~&~$\frac{ |{\Pi_\rhoo(L)}-{\Pi_\rhoo(L_s)}|}{\Pi_\rhoo(L_s)} \times 100~$  \\  
	\hline
	\hline
      Squared $\mathcal H_2$-Norm: $\frac{1}{2\beta} \sum_{i=2}^n{\lambda_i^{-2}}$& 17.58 \%\\ 
      \hline
      Local Deviation Error:  $\frac{1}{2 \beta}\sum_{i=1}^n {d_i}^{-2}$ & 11.38 \% \\
      \hline
      \end{tabular}}
      \caption{\small{Relative performance loss percentage of network \eqref{formation-1}-\eqref{formation-2} with coupling graph Fig. \ref{ex_3}(a) with respect to its abstraction Fig. \ref{ex_3}(b).}}\label{table_5}
\end{table}

\end{example}

\section{Discussion}

{We have introduced a notion of {abstraction} for  a class of linear consensus networks based on notions of spectral sparsification. There have been  several close-in-spirit notions of graph sparsifications  in the context of theoretical computer science.  While these other notions are interesting for their own sake from a combinatorial standpoint, their connections to performance analysis and synthesis of dynamical networks are not trivial and require further scrutiny. In this context,  for a given graph there are several ways to define sparse subgraphs, namely, 
 distance sparsifiers that approximate all pairwise distances up to a multiplicative and/or additive error (see \cite{Spanner} and subsequent research on spanners),
{cut sparsifiers} that approximate every cut to an arbitrarily small multiplicative error \cite{Benczur},
{spectral sparsifier} that approximate every eigenvalue to an arbitrarily small multiplicative error \cite{Spielman:2004}, and many more.
It is shown that sparsifiers can be constructed by sampling links according to their strength, effective resistance \cite{Spielman}, edge connectivity \cite{Fung}, or by sampling random spanning trees \cite{Goyal}.
Bencz\'ur and Karger propose a randomized algorithm to construct a cut sparsifier in $\mathcal O(m \log^2 n)$ time for unweighted graphs and $\mathcal O(m \log^3 n)$ time for weighted graphs \cite{Benczur,Fung}. The notion of spectral sparsifier is stronger than cut sparsifier, which implies spectral sparsifiers are also cut sparsifiers. One of our current research directions is to investigate all these methods of sparsification and explore their connections to dynamical networks.  }

In our analysis, we assumed that the input matrix of the network is an identity matrix. Further extension of the current work includes linear consensus networks over directed coupling graphs with arbitrary input matrices. However, this generalization is challenging and even the problem formulation demands some appropriate conditions to guarantee the boundedness of  performance measures.

\section*{Acknowledgment}
The authors would like to thank Yaser Ghaedsharaf for suggesting materials of Section \ref{ell0-ell1} and Professor Victor M. Preciado for several fruitful discussions about the subject.
\begin{spacing}{1}
\bibliography{main_Milad}
\end{spacing}

\begin{IEEEbiography}{Milad Siami} 
received his dual B.Sc. degrees in electrical engineering and pure mathematics from Sharif University of Technology in 2009, M.Sc. degree in electrical engineering from Sharif University of Technology in 2011, and M.Sc. and Ph.D. degrees in mechanical engineering from Lehigh University in 2014 and 2017 respectively. From 2009 to 2010, he was a research student at the Department of Mechanical and Environmental Informatics at the Tokyo Institute of Technology, Tokyo, Japan. He is currently a postdoctoral associate in the Institute for Data, Systems, and Society at MIT. His research interests include distributed control systems, distributed optimization, and applications of fractional calculus in engineering. He received a Gold Medal of National Mathematics Olympiad, Iran (2003) and the Best Student Paper Award at the 5th IFAC Workshop on Distributed Estimation and Control in Networked Systems (2015). Moreover, he was awarded  RCEAS Fellowship (2012), Byllesby Fellowship (2013), Rossin College Doctoral Fellowship (2015), and Graduate Student Merit Award (2016) at Lehigh University.
\end{IEEEbiography}

\begin{IEEEbiography}{Nader Motee}
(S'99-M'08-SM'13) received his B.Sc. degree in Electrical Engineering from Sharif University of Technology in 2000, M.Sc. and Ph.D. degrees from University of Pennsylvania in Electrical and Systems Engineering in 2006 and 2007 respectively. From 2008 to 2011, he was a postdoctoral scholar in the Control and Dynamical Systems Department at Caltech. He is currently an Associate Professor in the Department of Mechanical Engineering and Mechanics at Lehigh University. His current research area is distributed
dynamical and control systems with particular focus on issues related to sparsity, performance, and robustness. He is a past recipient of several awards including the 2008 AACC Hugo Schuck best paper award, the 2007 ACC best student paper award, the 2008 Joseph and Rosaline Wolf best thesis award, a 2013 Air Force Office of Scientific Research  Young Investigator Program award (AFOSR YIP), a 2015 NSF Faculty Early Career Development (CAREER) award, and a 2016 Office of Naval Research Young Investigator Program award (ONR YIP).
\end{IEEEbiography}

\end{document}